\documentclass[aps,prr,superscriptaddress,twocolumn,showpacs,longbibliography]{revtex4-2}

\usepackage{physics}
\usepackage[utf8]{inputenc}
\usepackage{amssymb,amsmath,amsfonts,dsfont,mathrsfs,amsthm,mathtools}
\usepackage{graphicx}
\usepackage{hyperref}
\hypersetup{
    colorlinks,
    citecolor=blue,
    filecolor=black,
    linkcolor=blue,
    urlcolor=black
}

\newcommand{\R}{\mathbb{R}}
\newcommand{\C}{\mathbb{C}}
\newcommand{\Z}{\mathbb{Z}}

\newcommand{\cC}{\mathcal{C}}

\newcommand{\ketr}[1]{#1 \rangle}

\DeclareMathOperator{\range}{range}

\newcommand{\comment}[1]{}

\usepackage{xcolor}
\definecolor{purp}{RGB}{160, 32, 240}

\newtheorem{theorem}{Theorem}[section]
\newtheorem*{conjecture*}{Conjecture}
\newtheorem{lemma}[theorem]{Lemma}
\newtheorem*{lemma*}{Lemma}

\newtheorem*{proposition*}{Proposition}

\theoremstyle{definition}
\newtheorem{remark}[theorem]{Remark}
\newtheorem{definition}[theorem]{Definition}

\begin{document}

\title{Erasure conversion in Majorana qubits via 
local quasiparticle detection}
\author{Abhijeet Alase}
\affiliation{Institute for Quantum Science and Technology, University of Calgary, 2500 University Drive NW, Calgary, Alberta T2N 1N4, Canada}
\affiliation{Centre for Engineered Quantum Systems, School of Physics, 
University of Sydney, Sydney, New South Wales 2006, Australia}
\author{Kevin D. Stubbs}
\affiliation{Department of Mathematics, University of California, Berkeley, California 94720 USA}
\author{Barry C. Sanders}
\affiliation{Institute for Quantum Science and Technology, University of Calgary, 2500 University Drive NW, Calgary, Alberta T2N 1N4, Canada}
\author{David L. Feder}
\affiliation{Institute for Quantum Science and Technology, University of Calgary, 2500 University Drive NW, Calgary, Alberta T2N 1N4, Canada}

\begin{abstract}
Quasiparticle poisoning errors in Majorana-based qubits are not suppressed 
by the underlying topological properties, which undermines the usefulness of
this proposed platform. This work tackles the errors originating from 
intrinsically excited quasiparticles by developing an
erasure conversion scheme based on local quasiparticle detection. 
To model such measurements, we begin by constructing 
the quasiparticle position operator for the Kitaev chain.
A measurement probe coupling to this operator is shown to allow projective measurements in
the Wannier quasiparticle basis. 
Detection of quasiparticles in a region of width $d$ adjacent to each Majorana zero-energy mode 
allows implementation of an error-detecting 
Majorana stabilizer code $\mathcal{C}_d$ based on microscopic fermionic (non-topological)
physical degrees of freedom. The implementation of $\mathcal{C}_d$ converts a 
large fraction of 
Pauli errors to erasure errors, thus achieving `erasure conversion' in Majorana qubits.
We show that the fraction of Pauli errors escaping conversion to erasure 
errors is exponentially small in $d$,
a result tied to the exponential localization of Wannier functions which we prove rigorously. 
The suppression in Pauli error rate comes at the cost of the 
erasure rate increasing sublinearly with $d$, but this can be
readily compensated for by a suitable outer code, with the net effect being 
a higher threshold rate of quasiparticle poisoning. The framework developed
here serves as a basis for understanding how realistic measurements, such as 
conductance measurements, could be utilized for achieving fault tolerance in
these systems. 
\end{abstract}

\maketitle

\section{Introduction}

Perhaps the biggest challenge to realizing fault-tolerant quantum computers is 
suppressing error rates in gate operations to below the thresholds dictated by 
state-of-the-art error-correcting codes. This was the main attraction of 
topological qubits from the point of view of fault-tolerant quantum 
computation~\cite{freedman2}: topological qubits provide a route to suppress 
errors at the hardware level. Information in topological qubits is encoded 
non-locally in low-energy anyonic excitations~\cite{Kitaev,kitaev03}.
Therefore, only non-local errors can alter the state of the stored qubit, thus 
suppressing Pauli errors. Whereas local errors can transfer population from the 
topologically degenerate ground states to higher energy excited states causing 
leakage, these processes are suppressed by constant bulk energy gap. At 
equilibrium, the density of excited higher-energy quasiparticles (QPs) is 
exponentially small in the ratio of the bulk energy gap to the absolute 
temperature~\cite{freedman2}.

The field of topological quantum computation began with the hope that one might
be able to perform arbitrarily long computations without requiring active error 
correction~\cite{Kitaev,freedman}. However, the topological systems that are 
currently in contention for being used as topological qubits lack self 
error-correcting properties~\cite{disorderassisted}, because they allow 
propagation of anyonic excitations at a constant energy 
cost~\cite{dennis02,pastawski10}. Topological qubits decohere due to the
interactions between a computational anyon and excited-state QPs, a process 
called QP poisoning. 
Whereas the density of higher-energy QPs is 
exponentially small in the bulk gap at thermal equilibrium, they can arise 
during non-adiabatic gate operations~\cite{karzig21}. Classical noise sources 
at characteristic frequencies close to the bulk gap can also lead to an 
increased density of QPs~\cite{Alicea20}. Unlike the errors that originate from
finite overlap of computational anyons, the QP poisoning errors cannot be 
suppressed by simply increasing the distance between computational anyons. 
Therefore, QP poisoning errors in topological qubits must be addressed using
other approaches.

Majorana-based qubits are leading candidates for topological 
qubits~\cite{Sarma2015}. Such a qubit can be realized on a tetron device, 
comprising a pair of semiconducting nanowires with strong spin-orbit coupling 
in close proximity to a 
superconductor~\cite{freedman,disorderassisted,Alicea20}. A tetron device hosts 
four spatially separated Majorana zero-energy modes (MZMs), and therefore has 
four degenerate ground states. 
A qubit is conventionally encoded in the two even-parity ground states of a 
tetron, and all logical operators are products of two out of the four MZMs 
hosted by the tetron. Since all four MZMs are spatially separated, the rates of 
logical Pauli errors are suppressed exponentially in the length of the 
nanowires as long as the noise processes are local~\cite{Knapp2018b}. 
Even-parity local noise processes can generally excite QPs, however. Sources 
of such noise processes include charge noise, thermal noise, and non-adiabatic 
gate operations~\cite{Knapp2018b}. The interaction of MZMs with these QPs, 
namely QP poisoning, leads to decoherence of the qubit stored in MZMs. 
As these even-parity processes can occur without any exchange of 
electrons with the environment, we refer to them as intrinsic QP 
poisoning processes to distinguish them from extrinsic QP 
poisoning caused by interactions with the environment~\cite{diego12}. Mitigating 
the influence of these intrinsic QPs is the focus of the current work.

QP poisoning errors can be corrected, in general, by implementing conventional 
error-correcting codes~\cite{Knapp2018a}. Specialized error-correcting codes 
that use computational anyons as physical degrees of freedom to store logical 
qubits have also been explored~\cite{dauphinais17}. Examples of such codes 
include Majorana surface codes~\cite{Vijay2015, mclauchlan22}, which belong to 
the larger family of Majorana fermion codes~\cite{majoranacode1}. 
However, achieving fault-tolerance with any 
error-correcting code requires physical error rates to be below a certain 
threshold value, a condition that topological qubits with high rates of QP 
poisoning may not meet. Moreover, the resources required for error correction 
depend crucially on physical error rates. For these reasons, strategies that 
further suppress the error rates at the hardware level are highly desirable.

In principle, given total control over the microscopic fermionic degrees of freedom,
the local QP poisoning errors in a Majorana qubit are expected to be detectable and 
correctable, for reasons that are reviewed in \S\ref{sec:bgec}.  
Yet, a high level of control over microscopic degrees of freedom is 
experimentally challenging in these systems, and would seem to defeat the purpose 
of using topological qubits in the first place. The problem addressed in this 
work is how to tackle QP poisoning errors with minimal 
access to the microscopic degrees of freedom.

Towards this goal, we explore the possibility of erasure conversion in Majorana qubits.
Erasure conversion is the process of converting Pauli errors to erasure errors.
This technique has already been demonstrated to be effective 
in some non-topological qubit 
platforms~\cite{chou2023, levine2023, wu2022, kang2023}. An error process is convertible
to erasure if it leaves a signature that can be detected without disturbing the qubit. 
In Majorana qubits, intrinsic QP poisoning processes in fact do leave such a 
signature. Consider a pair of QPs excited in a topological wire, one of which 
is absorbed by a computational anyon, causing a QP poisoning error. This leaves 
behind a QP excitation in the bulk of the wire away from the MZMs. 
Detection of this QP without affecting the stored qubit should be possible,
because the qubit is stored in MZMs whereas the QP resides in the bulk of the 
wire. The QPs that are excited near the ends of the wire are more likely to 
poison the MZMs than those excited away from the ends, so a local detection of 
QPs should be much more effective than a position non-resolving detection 
scheme for the purpose of erasure conversion. 

There are three main challenges to formalizing this strategy. 
The first and most important is to model the local detection of QPs. 
In the Kitaev tetron model of a Majorana qubit the description of QPs is 
readily available, as the Hamiltonian is quadratic and diagonalizable via a
Bogoliubov transformation. However, these QPs are not localized except 
at the fixed point parameter values. To define local QP measurements,
we first construct a QP position operator that satisfies certain desiderata. 
The eigenstates of the QP position 
operator describe Wannier QPs~\cite{kivelson82}, which we prove are 
exponentially localized throughout the topological phase. This localization 
property is crucial to achieving high levels of suppression of Pauli errors.
Local QP detection is then defined as a projective measurement
in the basis of Wannier QPs.

Having modeled the QP detectors, the next task is to 
devise an erasure conversion scheme based on QP detection.
In this work, we make a crucial observation that an erasure
conversion scheme is equivalent to implementing an
error-detecting code at the hardware level, and construct some 
simple error-detecting Majorana stabilizer codes that can be 
implemented with QP detectors. The first application is to the fixed-point 
parameter values of the Kitaev tetron, as the ground states at these parameter 
values are known to form a Majorana fermionic stabilizer 
code~\cite{majoranacode1}. We show that the stabilizers of this code are in 
fact the Wannier QP parity operators. This observation allows for the design of 
error-detecting codes for encoding a single qubit in the physical fermionic 
degrees of freedom of the Kitaev tetron throughout the topological phase, such 
that any local parity-preserving error taking the system out of the code 
space can be detected by QP measurement. More concretely, we construct a family 
of codes $\{\mathcal{C}_d\}$ parametrized by a non-negative integer $d<n/2$, 
where $n$ is the length of each chain in Kitaev tetron. Effectively, 
our erasure conversion scheme has a tunable integer parameter $d$.
The even-parity distance of each code $\mathcal{C}_d$ is $4d+4$. 
We show that a QP detector with inverse 
resolution $\lambda$ enables implementation of the codes 
$\{\mathcal{C}_d,\ \lambda \le d\}$, where $\mathcal{C}_0$ is identical to the 
conventional encoding of the qubit in a tetron at the fixed point. 
Wannier function-based error-detecting codes in this work are
conceptually similar to the smeared stabilizer codes investigated 
in Ref.~\cite{hastings}.

There are some key differences between the aforementioned 
error-detection codes and other Majorana codes in the literature,
such as the Majorana surface code. The physical degrees of freedom in our
codes are MZMs as well as the bulk fermion modes in the tetron, 
which permits the implementation of error-detection codes on a 
single Majorana tetron via QP detection. In contrast, for most Majorana codes 
the physical degrees of freedom are MZMs hosted by constituent tetrons/hexons,
implemented on a network comprising multiple tetrons or hexons, which do not 
require QP detection but rather joint parity measurements involving multiple 
MZMs. We emphasize that the error-detecting codes in this 
work do not provide a substitute for higher-level Majorana codes; rather, 
they are envisaged as forming a lower layer of error correction on 
which other (erasure-tolerant) codes can be concatenated, as is the case for erasure conversion schemes for other qubit platforms.

To assess the performance of the erasure conversion scheme, this
work considers a simple model of intrinsic QP poisoning. 
We show that the fraction of Pauli errors escaping conversion to erasure 
errors is exponentially small in~$d$. On the other hand, the erasure errors
increase sublinearly in $d$. Erasure errors are known to be 
easier to correct than Pauli errors~\cite{Stace2009,Stace2010}. This is evident 
from the fact that some quantum error correction codes achieve threshold erasure rates 
of 50\% per QEC cycle, whereas the threshold Pauli error rates are usually on 
the order of 1\%. Consequently, QP detection allows
fault-tolerant computation for significantly higher value of QP poisoning rates
compared to no QP detection. Moreover, this implies a significant reduction
in resources required for implementing error correction. 
 
The organization of this manuscript is follows. In \S\ref{sec:bg}, the 
relevant properties of the Kitaev tetron are reviewed. A model of QP detector 
with finite spatial resolution is constructed in \S\ref{sec:qpdetect}. In 
\S\ref{sec:edcodes}, we design Majorana error-detecting codes that can be 
implemented with the help of QP detectors.
These codes provide a rigorous description of the erasure conversion scheme. 
The performance of our erasure conversion scheme is 
assessed in \S\ref{sec:performance}. Finally, the results are summarized and 
open questions are discussed in \S\ref{sec:conclusion}.

\section{Background}
\label{sec:bg}
In this section, we review the theory of quadratic number non-conserving 
fermionic Hamiltonians, and the Kitaev chain model of a one-dimensional 
topological superconductor and its error-correcting properties. We also briefly 
review some tools from the covariance matrix framework. 

\subsection{Quadratic fermionic operators}
Let $a_j$ and $a^\dagger_j$, $j=1,\dots,n$, be fermionic annihilation and 
creation operators respectively on the Fock space $\mathcal{F}_n$. They satisfy
the canonical anticommutation relations
\begin{equation}    
[a_j,a^\dagger_{j'}]_+ = \delta_{jj'}, \quad [a_j,a_{j'}]_+ = 0,
\end{equation}
where $[\bullet,\bullet]_+$ denotes the anticommutator. The fermion-number and 
total-parity operators are defined to be $\hat{N} = \sum_j a^\dagger_j a_j$
and $(-1)^{\hat{N}}$ respectively. A quadratic fermionic Hermitian operator 
$\hat{A}$ can be expressed as~\cite{blaizot}
\begin{equation}
\label{quadraticoperator}
    \hat{A} = \frac{1}{2}\Phi^\dagger A \Phi + \text{constant},
\end{equation}
where $\Phi^\dagger = [a_1^\dagger \  \dots  \ a_n^\dagger \ a_1 \ \dots \ a_n]$
and $A$ is the Bogoliubov de-Gennes (BdG) matrix. The size of $A$ is 
$2n \times 2n$ and we denote its entries by $A_{jm,j'm'}$, where 
$j,j' \in {1,\dots,n}$ and $m,m'\in\{1,2\}$. Such an operator is not 
necessarily number conserving, i.e., $[\hat{A},\hat{N}]_- \ne 0$, where 
$[\bullet,\bullet]_-$ denotes the commutator. However, $\hat{A}$ conserves the
total parity, i.e., $[\hat{A},(-1)^{\hat{N}}]_- = 0$. The BdG matrix $A$ 
represents the action of commutator of $\hat{A}$ with linear fermionic 
operators, i.e.
\begin{align}
    &&[\hat{A},a^\dagger_j]_- = \sum_{j'=1}^{n}\left(A_{j1,j'1}a^\dagger_{j'}
    + A_{j1j'2}a_{j'}\right), \nonumber\\
    &&[\hat{A},a_j]_- = \sum_{j'=1}^{n}\left(A_{j2,j'1}a^\dagger_{j'}
    + A_{j2,j'2}a_{j'}\right).
\end{align}
Let $\mathcal{H}_{\rm BdG} = \text{span }\{a_j, a^\dagger_j\}$ denote the 
vector space of linear fermionic operators on $\mathcal{F}_n$. Let the basis of 
$\mathcal{H}_{\rm BdG}$ be denoted as $\{\ket{j,1} = a^\dagger_j,\ 
\ket{j,2} = a_j\}$; the kets signify membership in $\mathcal{H}_{\rm BdG}$. 
Throughout this work, small English or Greek alphabets in kets are used for 
vectors in $\mathcal{H}_{\rm BdG}$, and capital English or Greek alphabets in 
kets denote states in $\mathcal{F}_n$.

Define an operator $\tau_x$ that maps particles to holes and vice versa,
$\tau_x \ket{j,1} = \ket{j,2}$ and $\tau_x \ket{j,2} = \ket{j,1}$, i.e.\ that
corresponds to a Pauli-X operator acting on the particle and hole blocks.
$A$ satisfies the particle-hole constraint
\begin{equation}
\label{particlehole}
    \left(\tau_x{\mathcal K}\right)A\left(\tau_x{\mathcal K}\right)^{-1}
    =\tau_x{\mathcal K}A{\mathcal K}\tau_x=-A,
\end{equation}
where ${\mathcal K}$ denotes complex conjugation. Due to 
Eq.~\eqref{particlehole}, for every eigenvalue $\epsilon> 0$ of $A$ there 
also exists an eigenvalue $-\epsilon$. Let $\{\epsilon_k,\ k=0,\dots,n-1\}$ 
denote non-negative eigenvalues of $A$ in a non-decreasing order. Note that if 
zero is an eigenvalue of $A$, then its multiplicity $s_0$ is even, and exactly 
half of these are included in the set above: $\epsilon_k=0$ for 
$k=0,\dots,s_0/2-1$. 
Let $\{\ket{e_k},\ k=0,\dots,n-1\}$ be an orthonormal basis of eigenvectors
corresponding to eigenvalues $\{\epsilon_k,\ k=0,\dots,n-1\}$.
Then due to particle-hole constraint \eqref{particlehole}, 
an orthonormal basis for eigenvalues $\{-\epsilon_k,\ k=0,\dots,n-1\}$ 
can be chosen to be $\{\tau_x{\mathcal K}\ket{e_k},\ k=0,\dots,n-1\}$.
With such a choice of eigenvectors, $A$ can be expressed as
\begin{equation}
\label{Diagonal}
    A = \sum_{k}\epsilon_k\left(\ket{e_k}\bra{e_k} - 
    \tau_x{\mathcal K}\ket{e_k}\bra{e_k}{\mathcal K}\tau_x\right),\quad
		\epsilon_k \ge 0.
\end{equation}
Diagonalization of the BdG operator $A$ then casts $\hat{A}$ into its QP form
\begin{equation}
\label{Normal}
    \hat{A} = \sum_{k} \epsilon_k e_k^\dagger e_k + \text{constant},
\end{equation}
where $e_k^\dagger$ and $e_k$ are creation and annihilation operators 
represented by $\ket{e_k}$ and $\tau_x\mathcal{K}\ket{e_k}$, respectively.

\subsection{Majorana zero-energy modes}
With $\hat{A}$ in Eq.~\eqref{quadraticoperator} the Hamiltonian of the system 
under consideration, the $e_k^\dagger$ and $e_k$ are QP creation and 
annihilation operators, and $\epsilon_k\geq 0$ are the QP energies. If the
smallest eigenvalue $\epsilon_0 > 0$, then the Hamiltonian has a unique ground 
state $\ket{\Omega} \in \mathcal{F}_n$ specified by the conditions
$e_k\ket{\Omega}=0$ for all $k$. The state $|\Omega\rangle$ can be interpreted 
as the QP vacuum, and $e_k^{\dag}|\Omega\rangle$ is a QP basis function.

Consider Kitaev's model of a one-dimensional p-wave 
superconductor~\cite{Kitaev}, whose Hamiltonian is given by
\begin{eqnarray}    
\label{KitaevModel}
\hat{H}_{\rm Kit}&=&-\mu\sum_{j=1}^n\left(a_j^\dag a_j - \frac{1}{2}\right)
\nonumber \\
&+&\sum_{j=1}^{n-1}\left(-wa_j^\dag a_{j+1}+\Delta a_j a_{j+1}+\text{H.c.}
\right),
\end{eqnarray}
where $\mu$ is the chemical potential, $w$ is the hopping amplitude and 
$\Delta$ is the superconducting pairing amplitude, and H.c.\ corresponds to
the Hermitian conjugate. We assume for simplicity that $\mu$, $w$ and 
$\Delta \ge 0$ are real numbers. Diagonalizing yields
\begin{equation}
    \hat{H}_{\rm Kit} = \sum_{k=0}^{n-1} \epsilon_k e^\dagger_k e_k.
\end{equation}

The parameter regime $|\mu| < |w|$ marks the topological phase of the Kitaev 
Hamiltonian. For fixed parameters $\mu,w,\Delta$ in this regime, there are two
nearly degenerate ground states $|\Omega_0\rangle$ and $|\Omega_1\rangle$ with 
even and odd fermionic parity, respectively, separated by an energy 
$\epsilon_0$ that is exponentially small in $n$~\cite{Kitaev}.
In the limit of large $n$, the operators $e_0,e_0^\dagger$ represent the 
annihilation and creation operators for the zero-energy modes. Furthermore, the 
global phase of $e_0$ (and consequently of $e_0^\dagger$) can be chosen such 
that the Majorana operators for these,
\begin{equation}
    \gamma_1 = e_0^\dagger + e_0 \text{ and } \gamma_2 = i(e_0^\dagger - e_0),
\end{equation}
are exponentially localized on the left and the right edges of the chain 
respectively, in the sense that
\begin{equation}
\begin{split}
    |\expval{j,1|\gamma_1}|\approx e^{-\kappa_1 j}, \quad  
    |\expval{j,2|\gamma_1}| \approx e^{-\kappa_2 j},\\
    |\expval{j,1|\gamma_2}|\approx e^{-\kappa_1 (n-j)}, \quad  
    |\expval{j,2|\gamma_2}| \approx e^{-\kappa_2 (n-j)}
\end{split}
\end{equation}
for some positive constants $\kappa_1,\kappa_2$~\cite{Kitaev,PRL}. Because 
$\epsilon_0 \approx 0$, the Majorana operators commute with the system 
Hamiltonian up to an exponentially small correction, 
$[\hat{H}_{\rm Kit},\gamma_1]_- \approx [\hat{H}_{\rm Kit},\gamma_2]_-
\approx 0$. Therefore, $\gamma_1$ and $\gamma_2$ are called the MZMs of 
$\hat{H}_{\rm Kit}$. The even total-parity ground state $\ket{\Omega_0}$ is 
related to the odd total-parity ground state $\ket{\Omega_1}$ via MZMs as 
$\ket{\Omega_1} \propto \gamma_1\ket{\Omega_0} \propto \gamma_2\ket{\Omega_0}$. 
The parity operator for the zero-energy mode $(-1)^{e_0^\dagger e_0}$ can also 
be expressed as $-i\gamma_1\gamma_2$. Therefore, the two almost degenerate 
ground states can be labeled by $+1$ and $-1$ eigenvalues of 
$i\gamma_1\gamma_2$. These states are also eigenstates of the total parity 
operator~$(-1)^{\hat{N}}$.

\subsection{Encoding a qubit in four Majorana zero-energy modes}
\label{sec:bgenc}
In principle, a qubit could be stored in the two degenerate ground states of a 
Kitaev chain. However, the parity superselection rule prohibits coherent 
superposition of the two ground states. Therefore, it is standard to use two 
chains for encoding one qubit, called a tetron configuration. The Hamiltonian 
of the tetron is defined in terms of operators 
\begin{equation}
    \{a_{p,j}, a^\dagger_{p,j},\ p=1,2,\ j=1,\dots,n\}
\end{equation}
on the Fock space $\mathcal{F}_{2n}$,
where $p$ labels the two chains in the tetron, and 
\begin{equation}
    [a_{p,j},a^\dagger_{p,j'}]_+ = \delta_{pp'}\delta_{jj'}, 
		\quad [a_{p,j},a_{p',j'}]_+ = 0.
\end{equation}
By ignoring the terms connecting the two nanowires~\cite{Alicea20}, and by 
modeling each chain by the Kitaev chain Hamiltonian, one obtains
\begin{equation}
\label{tetronHam}
    \hat{H}_{\rm tet} = \hat{H}_{1,{\rm Kit}} + \hat{H}_{2,{\rm Kit}},
\end{equation}
where $\hat{H}_{p,{\rm Kit}}$ for $p=1,2$ describes the Hamiltonian of the 
$p$th chain and is obtained from Eq.~\eqref{KitaevModel} by replacing each 
$a_j$ ($a^\dagger_j$) by $a_{p,j}$ ($a^\dagger_{p,j}$), each of which act on 
the Fock space of the tetron $\mathcal{F}_{2n}$. For simplicity, assume 
that the parameters $(\mu, w, \Delta)$ for the two chains in the tetron are 
identical, in which case the Kitaev tetron has a four-dimensional ground space
when $|\mu|<|w|$. The logical qubit can be intialized in 
$\ket{\bar{0}}, \ket{\bar{1}}$ states by preparing the system in the 
ground states $\ket{\Omega_{00}},\ket{\Omega_{11}} \in \mathcal{F}_{2n}$ 
respectively, both of which are $+1$ eigenstates of the MZM-parity operator 
$(i\gamma_1\gamma_2)(i\gamma_3\gamma_4)$~\cite{diego12,activecorrection4,Knapp2018a,Smith2020,Lai2020}.
Logical single-qubit Pauli gates are given by operators  
\begin{equation}
\label{gates}
    \bar{X} \mapsto -i\gamma_1\gamma_3,\quad  
    \bar{Y} \mapsto -i\gamma_2\gamma_3,\quad 
    \bar{Z} \mapsto -i\gamma_1\gamma_2,
\end{equation}
and therefore can be implemented by manipulating the MZMs. 

Whereas preparing the ground states $\ket{\Omega_{00}},\ket{\Omega_{11}}$ 
amounts to initializing the qubit in the logical $\ket{\bar{0}},\ket{\bar{1}}$ 
states respectively, it would be slightly inaccurate to state that the qubit is 
encoded in the states $\ket{\Omega_{00}},\ket{\Omega_{11}}$. Rather, the state 
of the qubit is completely determined by the expectation values of the Pauli 
observables defined in Eq.\eqref{gates} together with the MZM-parity operator
$(i\gamma_1\gamma_2)(i\gamma_3\gamma_4)$, and does not depend on whether or not
the bulk (above-gap) QPs are excited. In other words, states orthogonal to the 
ground energy space nevertheless represent normalized logical qubit states as 
long as they are $+1$ eigenstates of the MZM parity operator. It is more 
accurate to state that the qubit is encoded in the even MZM-parity subspace 
$\mathcal{C}_0$ of a four-dimensional tensor factor of $\mathcal{F}_{2n}$, on 
which the MZMs act non-trivially~\cite{Knapp2018b}. In fact, this 
four-dimensional space carries a representation of the algebra generated by 
$\{\gamma_1,\gamma_2,\gamma_3,\gamma_4\}$. This point is revisited in more 
detail in \S\ref{sec:edcodesB}.

\subsection{Kitaev tetron as an error-correcting code}
\label{sec:bgec}

\begin{figure}
\includegraphics[width = \columnwidth]{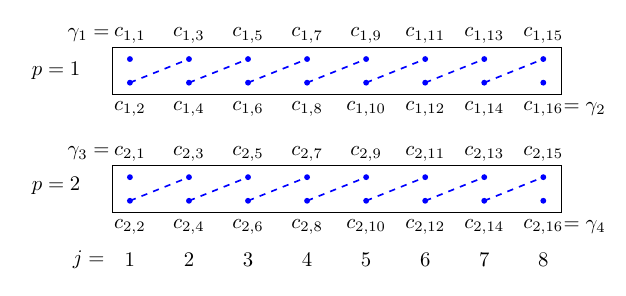}
\caption{Schematic of a tetron with length $n = 8$, comprising two 
Kitaev chains labeled by $p=1,2$ and outlined in black rectangles. 
The blue dots 
and blue dashed lines indicate 
Majorana operators $\{c_{p,j},\ p=1,2,\ j=1,\dots,n\}$ 
and pairing between adjacent Majorana operators
at the fixed point respectively. The four MZMs hosted by the tetron 
are denoted by $\{\gamma_1,\gamma_2,\gamma_3,\gamma_4\}$ \eqref{MZMs}.
\label{fig:tetron}}
\end{figure}

The ground states of the Kitaev tetron have interesting error-correcting 
properties~\cite{majoranacode1,disorderassisted}. To describe these properties, 
we first represent the Hamiltonian in Eq.~\eqref{tetronHam} in terms of 
Majorana operators $\{c_{p,j},\ j=1,\dots 2n\}$, which are defined by the 
relations
\begin{equation}
a_{p,j}=\frac{c_{p,2j-1} + i c_{p,2j}}{2} , \quad
a_{p,j}^\dagger=\frac{c_{p,2j-1} - i c_{p,2j}}{2} .
\end{equation}
The Majorana operators are self-adjoint, $c_{p,j}^\dag=c_{p,j}$, and obey 
commutation rules
\begin{equation}
[c_{p,j}, c_{p',j'}]_+ = 2\delta_{p,p'}\delta_{j,j'} \mathds{1}.
\end{equation}
The Hamiltonian in Eq.~\eqref{KitaevModel} can then be expressed as 
\begin{multline}
\label{eq:KitaevMajorana}
\hat{H}_{\rm Kit}=\frac{i(\Delta+w)}2 \sum_{j=1}^{n-1} c_{2j} c_{2j+1} \\
+ \frac{i(\Delta-w)}2 \sum_{j=1}^{n-1} c_{2j-1} c_{2j+2} 
- \frac{i}2 \sum_{j=1}^n \mu c_{2j-1} c_{2j}. 
\end{multline}
For the ``fixed point'' parameter values $\mu = 0$ and $w = \Delta = 1$, one
obtains
\begin{equation}
\label{eq:hamiltoniandefinition}
\hat{H}_{\rm Kit}=  \sum_{j=1}^{n-1} ic_{2j} c_{2j+1}.
\end{equation}
The Hamiltonian for the tetron can then be expressed as
\begin{equation}
\hat{H}_{\rm tet}=  \sum_{p=1,2}\sum_{j=1}^{n-1} ic_{p,2j} c_{p,2j+1}.
\end{equation}
Each chain of the Kitaev tetron hosts one MZM on each of its ends. 
Therefore, the tetron hosts four MZMs in total, labeled by 
\begin{equation}
\label{MZMs}
    \gamma_1 = c_{1,1},\quad \gamma_2 = c_{1,2n},\quad \gamma_3 = c_{2,1},\quad \gamma_4 = c_{2,2n}
\end{equation}
in Fig.~\ref{fig:tetron}.

\begin{figure}
\includegraphics[width = \columnwidth]{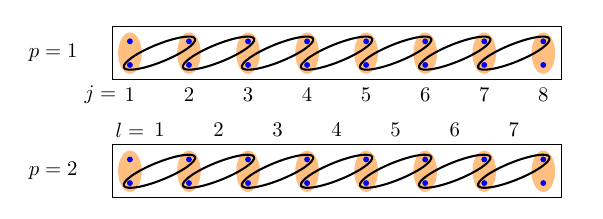}
\caption{Schematic of stabilizers and elementary errors for a tetron of 
		length $n=8$. The support of each stabilizer $Q_{p,l}$ is outlined in 
		black, and the support of each elementary error $E_{p,j}$ is shaded in 
		orange.}
\label{fig:stabilizers}
\end{figure}

The even-parity ground states of the Kitaev tetron are known to form a Majorana 
fermionic stabilizer code $\mathcal{C}_{\rm GS}$, which can in principle 
correct local parity preserving errors~\cite{NC,majoranacode1,disorderassisted}.
The explicit dependence of $\mathcal{C}_{\rm GS}$ on the length $n$ of the 
Kitaev tetron is omitted for brevity. The four-dimensional ground space of the 
Kitaev tetron is the subspace defined by the stabilizer generators (see 
Fig.~\ref{fig:stabilizers})
\begin{equation}   
\label{stab}
    \{Q_{p,l} = -ic_{p,2l}c_{p,2l+1},\quad p=1,2,\  l=1,\dots,n-1\},
\end{equation}
and the two even MZM-parity ground states are stabilized by the MZM-parity 
operator
\begin{equation}    
    Q_{\rm parity} = (-i\gamma_1\gamma_2)(-i\gamma_3\gamma_4) = -\gamma_1\gamma_2\gamma_3\gamma_4.
\end{equation}
In fact, the Hamiltonian of the tetron can be expressed in terms of the 
stabilizers as 
\begin{equation}
    \hat{H}_{\rm tet}=  -\sum_{p=1,2}\sum_{l=1}^{n-1} Q_{p,l}.
\end{equation}

Measurement of these stabilizers leads to discretization of even-parity noise 
processes, and the resulting discretized processes are products of elementary 
errors
\begin{equation}
\label{errors}
\{E_{p,j} = ic_{p,2j-1}c_{p,2j},\ p=1,2,\ j=1,\dots,n\}.
\end{equation}
Any arbitrary even-parity error $E$ can be expressed as a linear combination of
terms that are products of elementary errors and stabilizers in 
Eq.~\eqref{stab}, that is
\begin{equation}
    E = \sum_r E_r, \quad E_r  = \prod_{p=1,2}\prod_{j \in \mathcal{J}_{r,p}} E_{p,j} \prod_{l \in \mathcal{L}_{r,p}} Q_{p,l},
\end{equation}
where $\mathcal{J}_{r,p} \subset \{1,\dots,n\}$ and 
$\mathcal{L}_{r,p} \in \{1,\dots,n-1\}$. The weight of $E$, denoted by 
${\rm wt}(E)$, is two times the maximum number of elementary errors on a single 
chain present in the decomposition of any one of the terms in the linear 
combination, that is
\begin{equation}
    {\rm wt}(E) = 2\max_{r,p} |\mathcal{J}_{r,p}|.
\end{equation}
The code $\mathcal{C}_{\rm GS}$ can correct exactly any even-parity error $E$ 
with ${\rm wt}(E) \le n-1$.

\subsection{Covariance matrix framework}
\label{sec:covariance}
In this section, we recall some basic definitions and propositions in the 
covariance matrix formalism for fermionic Gaussian states~\cite{surace22}.
These propositions are used for simulating qubits encoded in certain 
error-detecting codes on a tetron. 
\begin{definition}
    The covariance matrix of a $\ket{\Psi} \in \mathcal{F}_{2n}$ is the antisymmetric matrix with entries
    \begin{multline}
    \label{covariance}
        M_{p,j;p',j'} = -\frac{i}{2}\expval{\Psi|[c_{p,j},c_{p',j'}]_-|\Psi},\\
        p,p' \in \{1,2\},\quad j,j'\in\{1,\dots,2n\}.
    \end{multline}
\end{definition}
\noindent Note that in Eq.~\eqref{covariance}, $(p,j)$ and $(p',j')$ denote the 
row and column indices of $M$, respectively.
\begin{lemma}
\label{lem:M0}
    For $\mu=0$ and $w=\Delta=1$, the covariance matrix $M_{\zeta}$ 
    of the ground state of $\hat{H}_{\rm tet}$ 
    satisfying $\expval{\Psi|(-ic_{p,1}c_{p,2n})|\Psi} = (-1)^\zeta$ with $\zeta \in \{0,1\}$,
    has only non-zero entries
    \begin{multline}   
    [M_{\zeta}]_{p,2j;p,2j+1} = 1, \quad [M_{\zeta}]_{p,1;p,2n} = (-1)^{\zeta},\\ 
    p = 1,2,\ j = 1,\dots,n-1.
    \end{multline}
\end{lemma}
The next lemma makes use of the Majorana operator basis 
$\{\ket{c_{p,j}/\sqrt{2}},\ p=1,2,\ j=1,\dots,2n\}$ of $\mathcal{H}_{\rm BdG}$, 
where
\begin{eqnarray}
    \ket{\frac{c_{p,2j-1}}{\sqrt{2}}} &=& \frac{\ket{p,j,1} + \ket{p,j,2}}{\sqrt{2}},\nonumber\\ 
    \ket{\frac{c_{p,2j}}{\sqrt{2}}} &=& -i\frac{\ket{p,j,1} - \ket{p,j,2}}{\sqrt{2}}.
\end{eqnarray}
\begin{lemma}
    Let $\hat{H}$ be a Hamiltonian quadratic in Majorana operators $\{c_{p,j}\}$, and let $H$ be 
    the corresponding BdG Hamiltonian expressed in the Majorana operator basis. If $M$ is the covariance
    matrix of $\ket{\Psi}$, then the covariance matrix of $e^{-i\hat{H}t}\ket{\Psi}$ is 
    $RMR^{\rm T}$ with $R = e^{iHt}$.
\end{lemma}
\begin{lemma}
\label{lem:wick}
Consider an ordered list of operators $L_1,\dots,L_{2m}$ such that every
operator $L_j$ is a linear combination of the Majorana operators $c_{p,1},\dots,c_{p,2n}$ 
with complex coefficients.
For $\ket{\Psi} \in \mathcal{F}_{2n}$, define an antisymmetric 
$2m \times 2m$ complex matrix $A$ such
that $A_{jj'} = \expval{\Psi|L_jL_{j'}|\Psi}$ for $j < j'$. Then
\begin{equation}
    \expval{\Psi|L_1\dots L_{2m}|\Psi} = \operatorname{Pf\,}(A),
\end{equation}
where $\operatorname{Pf\,}(\bullet)$ denotes the Pfaffian.
\end{lemma}

\section{Model of quasiparticle detection}
\label{sec:qpdetect}

We now construct a model for QP detection with finite spatial resolution. 
Finite spatial resolution is important for two reasons. First, such a model 
could be easier to realize in some systems. Second, a detector with high 
spatial resolution is potentially more advantageous for detecting local errors,
as is discussed later.

\subsection{Quasiparticle position operator}
A local measurement of QPs is conceivable if the measurement probe
can couple to the position of the QPs. Therefore, 
to model a detector with finite spatial resolution, one must first derive an
analog of the position operator for QP excitations. Define the position 
operator for fermions on one of the chains as
$\hat{X} = \sum_j ja^\dagger_j a_j$ (this is the position analog of the current 
operator). First, $\hat{X}$ is a quadratic operator, which means that it acts 
independently on the particles. Second, $\hat{X}$ commutes with the number 
operator $\hat{N} = \sum_j a_j^\dagger a_j$, which means that a projective 
measurement of $\hat{X}$ on a Fock state with a fixed number of particles does 
not alter the number of particles. In other words, 
$\hat{X}$ is a nondemolition measurement~\cite{kok07}.
Third, the BdG representation of $\hat{X}$, 
which is $X = \sum_j j(\ket{j,1}\bra{j,1}-\ket{j,2}\bra{j,2})$, satisfies 
$\expval{\psi|X|\psi} = \sum_j j|\psi_j|^2$ for any creation operator 
$\ket{\psi} = \sum_j \psi_j \ket{j,1} \in \mathcal{H}_{\rm BdG}$, 
and similarly satisfies 
$(\bra{\psi}\tau_x{\mathcal K})X(\tau_x{\mathcal K}\ket{\psi}) 
= -\sum_j j|\psi_j|^2$ for the corresponding annihilation operator.

Taking inspiration from the defining properties of $\hat{X}$, the QP position 
operator $\hat{X}_{\rm qp}$ should satisfy the following desiderata:
\begin{enumerate}
    \item $\hat{X}_{\rm qp}$ acts independently on QPs.
    \item A projective measurement of $\hat{X}_{\rm qp}$ is a non-demolition 
		measurement in the sense that for any eigenstate of $\hat{N}_{\rm qp}$, the 
		post-measurement state is also an eigenstate of $\hat{N}_{\rm qp}$ with the 
		same eigenvalue.
    \item If a QP is localized near site $j$, then 
		$\expval{\psi|X_{\rm qp}|\psi} \approx j$.
\end{enumerate}
For the first property to be satisfied, $\hat{X}_{\rm qp}$ must be quadratic in 
QP creation and annihilation operators $\{e^\dagger_k,e_k,\ k=1,\dots,n-1\}$.
While MZMs described by operators $e_0^\dagger,e_0$ are also technically QPs,
they are excluded from the set of QPs as they are not responsible for leakage
from the ground-state manifold. For the second property to be satisfied, 
$\hat{X}_{\rm qp}$ must commute with the QP number operator $\hat{N}_{\rm qp}$. 

To construct the operator $\hat{X}_{\rm qp}$, consider its action on the BdG 
space. In the following, $H$ is a Hermitian operator on $\mathcal{H}_{\rm BdG}$ 
satisfying the particle-hole constraint in Eq.~\eqref{particlehole}. 
Let $\{\ket{e_k},\ k=0,1,\dots,n-1\}$ be the eigenvectors of $H$ with 
non-negative eigenvalues $\{\epsilon_k,\ k=0,1,\dots,n-1\}$ as above.
Let
\begin{equation}
  \label{eq:riesz-p}
P_{\rm qp} = \sum_{k=1}^{n-1}\ket{e_k}\bra{e_k}
\end{equation}
be the projector on the space of QP creation operators, and 
\begin{equation}
P_{\rm \overline{qp}} = \left(\tau_x{\mathcal K}\right)P_{\rm qp}
\left({\mathcal K}\tau_x\right) = \sum_{k=1}^{n-1}
\left(\tau_x{\mathcal K}\right)\ket{e_k}\bra{e_k}
\left({\mathcal K}\tau_x\right)
\end{equation}
be the projector on the space of QP annihilation operators. Define 
$\tilde{X} = NX = \sum_{j}j\left(\ket{j,1}\bra{j,1}+\ket{j,2}\bra{j,2}\right)$.
\begin{theorem}
\label{thm:qpposition}
Let $P_{\rm qp}, \hat{N}_{\rm qp}$ and $\tilde{X}$ be as defined above. 
\begin{enumerate}
\item The BdG operator defined by 
\begin{equation}
    X_{\rm qp} = P_{\rm qp}\tilde{X}P_{\rm qp} 
		- P_{\rm \overline{qp}}\tilde{X}P_{\rm \overline{qp}},
\end{equation}
satisfies the particle-hole constraint, Eq.~\eqref{particlehole}.
\item The observables $\hat{X}_{\rm qp}$ and $\hat{N}_{\rm qp}$ can be observed 
    simultaneously, i.e., 
    \begin{equation}
        [\hat{X}_{\rm qp},\hat{N}_{\rm qp}]_- = 0.
    \end{equation}
\item For all $\ket{\psi} \in \text{span}\,\{\ket{e_k},\ k=1,\dots,n-1\}$,
\begin{align}
\label{qpposition}
    \expval{\psi|X_{\rm qp}|\psi} 
		&= -\bra{\psi}\left(\tau_x{\mathcal K}
		\right)X_{\rm qp}\left({\mathcal K}\tau_x\right)\ket{\psi}
		\nonumber\\
    &= \sum_j j\left(|\langle j,1|\psi\rangle|^2 
		+ |\langle j,2|\psi\rangle|^2\right).
\end{align}
\end{enumerate}
\end{theorem}

\begin{proof}
The particle-hole constraint follows from the fact that 
$\left(\tau_x{\mathcal K}\right)P_{\rm qp}\left({\mathcal K}\tau_x\right)
= P_{\rm \overline{qp}}$ and 
$\left(\tau_x{\mathcal K}\right)\tilde{X}\left({\mathcal K}\tau_x\right) 
= \tilde{X}$. The second statement, namely 
$[\hat{X}_{\rm qp},\hat{N}_{\rm qp}]_- = 0$, follows from the relations 
$N_{\rm qp} = P_{\rm qp}-P_{\rm \overline{qp}}$, 
$P_{\rm qp}^2 = P_{\rm qp}$, $P_{\rm \overline{qp}}^2 = P_{\rm \overline{qp}}$ 
and $P_{\rm qp}P_{\rm \overline{qp}} = 0$, where $N_{\rm qp}$ is the BdG
representation of $\hat{N}_{\rm qp}$. For the third statement, using 
$P_{\rm qp}\ket{\psi} = \ket{\psi}$ and $P_{\rm \overline{qp}}\ket{\psi} = 0$, 
one obtains $\bra{\psi}P_{\rm qp}\tilde{X}P_{\rm qp}\ket{\psi}
= \bra{\psi}\tilde{X}\ket{\psi}$. A straightforward calculation yields
\begin{align}
    \bra{\psi}\tilde{X}\ket{\psi}
    &= \bra{\psi}\tilde{X}\sum_{j=1}^{N}\left(\expval{j,1|\psi}\ket{j,1} 
		+ \expval{j,2|\psi}\ket{j,2}\right)\nonumber\\
    &= \bra{\psi}\sum_{j=1}^{N}\left(j\expval{j,1|\psi}\ket{j,1} 
		+ j\expval{j,2|\psi}\ket{j,2}\right) \nonumber\\
    &= \sum_j j \left(|\langle j,1|\psi\rangle|^2+|\langle j,2|\psi\rangle|^2
		\right).
\end{align}
\end{proof}

\subsection{Exponential localization of Wannier quasiparticles}
Next we show that the eigenmodes of $\hat{X}_{\rm qp}$ (represented by 
eigenvectors of $X_{\rm qp}$) are exponentially localized in space. These 
localized eigenmodes are the analogs of exactly localized position eigenmodes 
$\{a^\dagger_j, a_j\}$ for the fermionic particles. This localization property 
plays an important role in making the QP detection measurements very powerful 
for the purpose of suppressing logical errors.

Suppose $X_{\rm qp}$ has the spectral decomposition
\begin{equation}
    X_{\rm qp} = \sum_{l=1}^{n-1} x_{l}\left[\ket{\phi_l}\bra{\phi_l}
		-\left(\tau_x{\mathcal K}\right)\ket{\phi_l}\bra{\phi_l}
		\left(\tau_x{\mathcal K}\right)\right],
\end{equation}
with $x_l \ge x_{l'}$ if $l>l'$ and $x_l > 0$ for all $l$. 
Then 
\begin{equation}
    \hat{X}_{\rm qp} = \sum_{l=1}^{n-1} x_{l} {\phi}^\dagger_l {\phi}_l, 
\end{equation}
where ${\phi}^\dagger_l$ is the operator represented by $\ket{\phi_l}$.
Given that $\sum_{l=1}^{n-1}\ket{\phi_l}\bra{\phi_l} 
= \sum_{k=1}^{n-1}\ket{e_k}\bra{e_k} = P_{\rm qp}$, then
\begin{equation}
    \sum_{l=1}^{n-1} {\phi}^\dagger_l {\phi}_l = \sum_{k=1}^{n-1} e^\dagger_k e_k = \hat{N}_{\rm qp}.
\end{equation}
The states $\{\ket{\phi_l}\}$ represent Wannier QP (WQP 
henceforth) excitations of the system. 
Therefore, ${\phi}_l$ and ${\phi}_l^\dagger$ describe the annihilation and 
creation of the $l$th WQP. 
Note that the index $l$ above ranges from $1$ to $n-1$, and not to $n$, as the 
MZMs are deliberately excluded from the set of QPs while constructing 
$\hat{X}_{\rm qp}$. The same index $l$ is used to enumerate the WQPs and the 
stabilizers in Eq.~\eqref{stab}; indeed, WQPs and stabilizers are closely 
related, as discussed in \S\ref{sec:qp2st}.

While it is mathematically convenient to prove that the Wannier functions are 
exponentially localized by making use of an infinite system, this work is 
expressly focused on systems with boundary terminations. We therefore prove our 
results for long but finite Kitaev chains, for Hamiltonians that model a 
topological superconductor with weak disorder and that host only one MZM on 
each edge. Let $\{ \ket{j,m} : j \in \Z^+, m \in \{1,\dots,D\} \}$ denote the 
standard basis of $\ell^2(\Z^+ \times \mathds{C}^{D}) = \mathcal{H}_{\rm BdG}$. 
Next, let $H$ be a Hermitian operator on $\mathcal{H}_{\rm BdG}$ and suppose 
that $H$ has exponentially decaying hopping and pairing correlations in the 
sense that for some constants $(C, \kappa)$,
\begin{equation}
\label{eq:h-exp-bd}
|\bra{j,m} H \ket{j',m'}| \leq C e^{-\kappa | j - j' |}.
\end{equation}
Such an operator is said to be exponentially 
localized (see Def.~\ref{def:exploc} for a rigorous definition.)
Suppose that there exists a lower bound $\delta > 0$ on the energy gap,
i.e., $\sigma(H) \cap (0, \delta) = \emptyset$ where $\sigma(\bullet)$ 
denotes the spectrum of its argument.
This condition is equivalent to requiring a bulk energy gap at least $\delta$ 
and no localized energy modes other than MZMs with energy smaller than $\delta$.
\begin{lemma}
\label{lem:p-exp-loc}
There exist constants $(C', \kappa')$ so that 
\begin{equation}
\label{eq:p-exp-bd}
|\bra{j,m} P_{\rm qp} \ket{j',m'}| \leq C' e^{-\kappa' | j - j' |}.
\end{equation}
\end{lemma}
In other words, $P_{\rm qp}$ is exponentially localized.

\begin{theorem}
\label{thm:wannier-1d}
Suppose $P_{\rm qp}$ is an exponentially localized operator on 
$\ell^2(\Z^+ \times \mathds{C}^{D})$. Then
\begin{enumerate}
    \item The operator $X_{\rm qp}$ has purely discrete spectrum on 
		$\range{(P_{\rm qp})}$.
    \item There exist constants $(C_{\rm qp}, \kappa_{\rm qp})$ depending only of $P_{\rm qp}$, 
		such that if $\ket{\phi} \in \range(P_{\rm qp})$ satisfies the 
		eigenvalue equation 
    \begin{equation}
    X_{\rm qp} \ket{\phi} = x \ket{\phi}
    \end{equation}
    for some $x > 0$, then $\ket{\phi}$ is exponentially localized about the 
		point $x$ in the sense that
    \begin{equation}
    |\bra{j, m} \ketr{\phi} | \leq C_{\rm qp} e^{-\kappa_{\rm qp} | j - x |}
    \end{equation}
\end{enumerate}
\end{theorem}
\noindent The proofs of Lemma~\ref{lem:p-exp-loc} and 
Theorem~\ref{thm:wannier-1d} are provided in the Appendices~\ref{sec:p-exp-loc} 
and~\ref{sec:wannier-1d} respectively, and they utilize some preliminary 
results about exponentially localized operators proved in 
Appendix~\ref{sec:app}.

\begin{figure}
    \centering
    \includegraphics[width = \columnwidth]{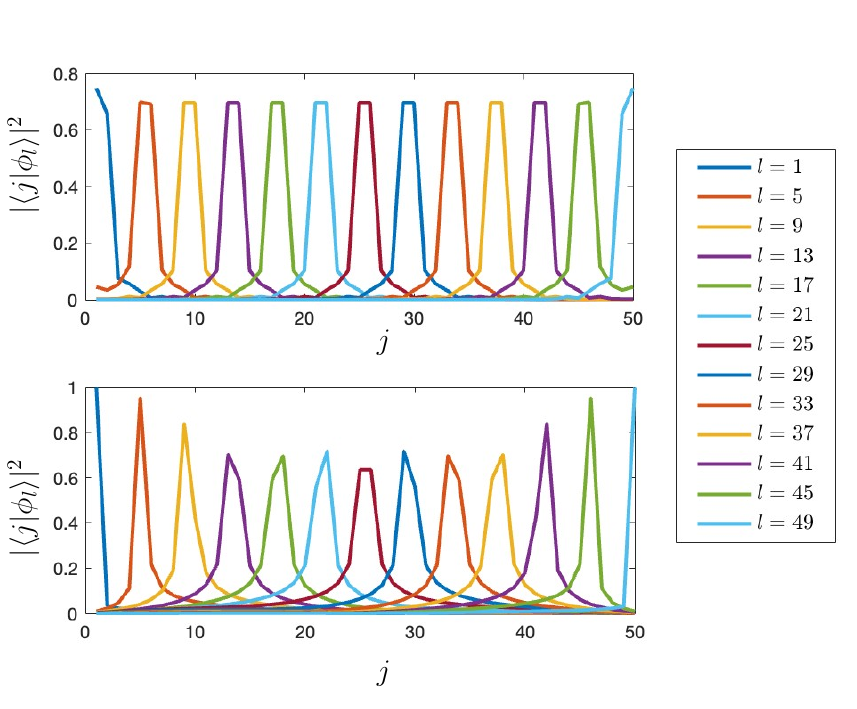}
    \caption{The amplitude of a few WQP excitations for parameters $\mu=0.3,w=1,\Delta = 0.4$ (top panel) and 
    $\mu=0.99,w=1,\Delta = 0.05$ (bottom panel). The tetron length is $n=50$ in both the plots.}
    \label{fig:wannier}
\end{figure}

Computing WQP operators numerically only requires the diagonalization of the 
operator $X_{\rm qp}$. The amplitude of some WQP operators are plotted in 
Fig.~\ref{fig:wannier}. Observe that the WQPs are tightly localized for 
parameters away from the topological phase boundary, Fig.~\ref{fig:wannier}
(top panel). The spatial extent increases as one approaches the topological 
phase boundary, Fig.~\ref{fig:wannier} (bottom panel), reflecting the closure 
of the many body gap.

\subsection{From Wannier quasiparticles to projective measurements}
A QP detector with inverse resolution $\lambda \in \mathds{Z}^+$ is 
characterized by the set of projective measurements that it can perform. 
For 
$l_{\min},l_{\max} \in \{1,\dots,n-1\}$ with $l_{\min} \le l_{\max}$, one can 
define the discrete interval 
\begin{equation}
    [[l_{\min},l_{\max}]]_p=
 \{(p,l_{\min}), (p,l_{\min}+1), \dots, (p,l_{\max})\}
\end{equation}
on the $p$th chain. Consequently, 
$[[l_{\min},l_{\max}]]_p \subseteq [[1,n-1]]_p$ for any valid 
$[[l_{\min},l_{\max}]]_p$.
Let
\begin{equation}
    \left|[[l_{\min},l_{\max}]]_p\right| = l_{\max}-l_{\min}+1
\end{equation}
be the length of the 
interval $[[l_{\min},l_{\max}]]_p$.
For any such 
$I_p = [[l_{\min},l_{\max}]]_p$ on the $p$th chain, the operator 
\begin{equation}
    \hat{N}_{\rm qp}^{I_p} = \sum_{(p,l) \in I_p} {\phi}^\dagger_{p,l} {\phi}_{p,l},
\end{equation}
counts the number of QPs in the interval $I_p$. 
Evidently, 
$\hat{N}_{\rm qp}^{[[1,n-1]]_p} = \sum_{k=1}^{n-1} e^\dagger_k e_k$. 
Note $l_{\max} \le n-1$ because we exclude the MZMs from the QPs,
and therefore we have only $n-1$ WQPs.

We now define an `on-off' detector measurement for QPs
in an interval $I_p$. Such a measurement yields $0$ if no QPs
are detected in the interval $I_p$, and yields $1$ if at least one 
QP is detected. Such a measurement is analogous to number non-resolving photon
measurement commonly used in quantum optics~\cite{kok07}.
Let us first introduce the heaviside function
on positive integers
\begin{equation}
    \theta(j) = \left\{\begin{array}{lcl}0 & \text{if} & j=0 \\
    1 & \text{if} & j>0
    \end{array}\right..
\end{equation}
Then 
\begin{equation}
    \theta(\hat{N}_{\rm qp}^{I_p}) = \mathds{1}-\prod_{(p,l) \in I_p} \phi_{p,l}\phi^\dagger_{p,l}
\end{equation}
is an orthogonal projector on the Fock 
states in which at least one QP is excited in the interval~$I$; 
this can be verified by observing that the second term on the right-hand side annihilates any such state.
For the interval $I_p$, we associate the   
projective measurement 
\begin{equation}
    \mathcal{M}_{I_p} = \{\theta(\hat{N}_{\rm qp}^{I_p}), \mathds{1}-\theta(\hat{N}_{\rm qp}^{I_p})\},
\end{equation}
which models a number non-resolving detection of the QPs in the interval $I_p$.
The measurement $\mathcal{M}_{I_p}$ can be performed by a detector with inverse 
resolution $\lambda$ if $|I_p| \ge \lambda$. Such a detector can also perform 
simultaneous measurements $\{\mathcal{M}_{I_{p,j}}\}$ as long as the intervals 
\{$I_{p,j}$\} are individually longer than $\lambda$ and are pairwise disjoint, 
that is $I_{p,j} \cap I_{p,j'} = \delta_{jj'}I_{p,j}$. Any interval on the 
$p=1$ chain is by definition disjoint with respect to any interval on the $p=2$ 
chain.

In summary, in quadratic fermionic systems such as the Kitaev tetron, the QP excitations 
are eigenmodes of the many-body Hamiltonian. A position-resolving detector is a 
device that approximately measures the position of a QP without destroying or 
creating any QPs. The eigenmodes of this operator are Wannier QPs, 
which are proven to be exponentially localized. 
A detector then determines if any Wannier QPs are excited in a 
given interval of length, which is bounded from below by the finite spatial 
resolution. 

This model of QP detection is used throughout to describe erasure conversion 
in Majorana qubits and for assessing its performance.
Note that an experimental probe that couples to the QP position operator can perform 
the QP detection measurements described above. In general, any detector that is sufficiently local 
and that detects the QPs only above the energy gap, i.e., excluding MZMs, 
will achieve the required measurements. Unfortunately, the latter requirement 
is not met by transport measurements, and therefore warrants further investigation.
Note also that the exponential localization of WQPs holds for generic models of one-dimensional
topological superconductors. Therefore, our results extend naturally to the corresponding
models of Majorana tetron qubits.

\section{Erasure conversion by quasiparticle detection}
\label{sec:edcodes}
We consider next the incorporation of general QP detection 
measurements in the error-correction framework. 
This is achieved by a process known as `erasure conversion',
in which additional capabilities at the hardware level are 
used to detect a set of errors that may have occurred 
on each qubit. Upon resetting or erasing the qubits on which errors are detected,
these errors turn into erasure errors~\cite{chou2023, levine2023, wu2022, kang2023}.
The main strength of erasure conversion is
that it can be made compatible with a concatenation of any outer
quantum error-correcting code with a suitable choice of 
decoder. The analysis is based on the observation that erasure conversion 
is equivalent to implementing an error-detecting code at the hardware level.
In the following, we first relate QP detection to certain collective 
stabilizer measurements and then construct error-detecting stabilizer codes 
that describe erasure conversion.

\subsection{From quasiparticle detection to stabilizer measurements}
\label{sec:qp2st}
We begin the construction of the required error-detecting codes
by first making a connection between the quasiparticle detectors and 
the stabilizers of the Kitaev tetron code at the fixed point.
The WQP operators $\{{\phi}_{p,l}\}$ at the fixed point ($\mu=0$, $w=\Delta=1$)
are straightforward to 
calculate, and are given by 
\begin{equation}    
    {\phi}_{p,l} = (c_{p,2l} - ic_{p,2l+1})/2, \quad l =1,\dots n-1,  
\end{equation}
with the corresponding eigenvalues being 
\begin{equation}
    x_{p,l} = l+1/2.
\end{equation}
Therefore
\begin{equation}
    (-1)^{{\phi}_{p,l}^\dagger {\phi}_{p,l}}
		= (-1)^{\hat{N}_{\rm qp}^{[[l,l]]_p}} 
		= Q_{p,l}, \quad l=1,\dots,n-1,
\end{equation}
where $[[l,l]]_p = \{(p,l)\}$ is a discrete interval of unit length on the 
$p$th chain. Thus, each of the stabilizers in Eq.~\eqref{stab} is 
in fact the parity operator for the corresponding WQP degree of freedom, at the
fixed point. The effect of each of the errors in Eq.~\eqref{errors} is to 
excite or de-excite either one (for $l \in \{1,n-1\}$) or two  (for $l \in \{2,\dots,n-1\}$) WQPs.

\begin{figure*}[t]
    \centering
    \includegraphics[width = 2\columnwidth]{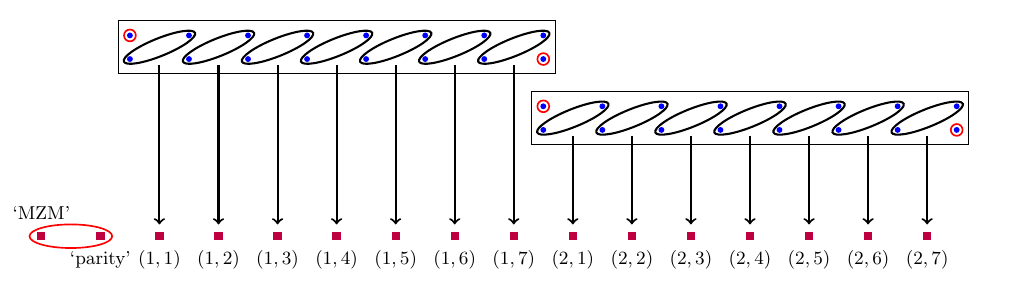}
    \caption{Schematic of the mapping from Majorana operators 
    (blue dots) to qubits (brown squares) 
    for a tetron of length $n=8$. The arrows show the mapping of $Q_{p,l}$ to $Z_{p,l}$.
    The MZMs (circles in red) map to operators on the `MZM' and the `parity' qubits outlined in red.}
    \label{fig:mapping}
\end{figure*}

\subsection{Quasiparticle detection-based error-detecting codes}
\label{sec:edcodesB}

In this section, we describe error-detecting codes that are implementable with 
a QP detector of a certain spatial resolution. We then define and construct 
some simple Majorana fermionic error-detecting codes with gauge degrees of 
freedom. The motivation for this is quite different from other subsystem codes 
in the literature, including the gauge color code~\cite{bombin2015}. In the 
latter, gauge degrees of freedom are fully accessible and are meant to 
facilitate the implementation of a universal transversal gate set. In contrast, 
the gauge degrees of freedom in this work are completely inaccessible; one
could describe the codes purely as stabilizer subspace codes by eliminating all 
gauge degrees of freedom from the description, but the gauge degrees are 
retained in this work as they permit a direct connection with the underlying 
tetron Fock space~$\mathcal{F}_{2n}$.

In the absence of QP detection, 
the qubit is encoded in the space $\mathcal{C}_0$ in a tetron, as explained in 
\S\ref{sec:bgenc}. This can be interpreted as an error-detecting stabilizer 
code of four MZMs, with the MZM-parity operator $(i\gamma_1\gamma_2)
(i\gamma_3\gamma_4)$ the only stabilizer. This code can detect any 
odd MZM-parity error. In this sense, the remaining $4n-4$ Majorana fermionic 
operators, which correspond to the bulk QP excitations on the two chains, are 
treated as gauge degrees of freedom~\cite{poulin05,dauphinais23}. Consequently, 
any process supported on the bulk QPs does not affect the state of the encoded 
qubit. Treating bulk QPs as gauge degrees of freedom is beneficial if detection 
of QPs is inaccessible, because the detection of odd MZM-parity errors alone is 
possible without requiring the detection of QPs. 
In contrast, the encoding in $\mathcal{C}_{\rm GS}$
discussed in \S\ref{sec:bgec} treats all bulk QPs as stabilizers; indeed, 
following the discussion in \S\ref{sec:qp2st}, the stabilizers measure the 
occupation parity of each bulk QP mode. Thus, in this picture there are no 
gauge degrees of freedom, and the encoding discussed in \S\ref{sec:bgec} 
could be implemented if some sort of QP detection were accessible. Furthermore,
this could be highly beneficial, as individual stabilizer measurements can 
detect any even total-parity error of weight $ \le 2n-2$.

We begin by exploring the subsystem decompositions associated with
the codes discussed in \S\ref{sec:bgenc} and \S\ref{sec:bgec}. A modified 
Jordan-Wigner transformation~$\nu$ can be used to map the Hilbert space of $2n$ 
qubits onto $\mathcal{F}_{2n}$. For reasons that will soon be obvious, label 
the first two qubits `MZM' and `parity', and the remaining $2n-2$ qubits by 
tuples $\{(p,l),\ p=1,2,\ l=1,\dots,n-1\}$. The invertible mapping is given by 
(see Fig.~\ref{fig:mapping})
\begin{widetext}
\begin{eqnarray}
\label{fermion2qubit}   
&Z_{\rm MZM} \xrightarrow{\nu} i\gamma_1\gamma_2, \quad X_{\rm MZM} 
\xrightarrow{\nu} i\gamma_2\gamma_3, \quad Z_{\rm parity} 
\xrightarrow{\nu} -\gamma_1\gamma_2\gamma_3\gamma_4, \quad X_{\rm parity} 
\xrightarrow{\nu} \gamma_4, \nonumber\\
&Z_{p,l} \xrightarrow{\nu} Q_{p,l}, \quad
X_{p,l} \xrightarrow{\nu} (-\gamma_1\gamma_2\gamma_3\gamma_4)(\prod_{l=1}^{n-1}Q_{1,l})^{\delta_{p,2}}(\prod_{l'=1}^{l-1}Q_{p,l'})(\phi^\dagger_{p,2l}+ \phi_{p,2l}),
\quad p=1,2,\quad l=1,\dots,n-1.
\end{eqnarray}
\end{widetext}
Similar to 
the standard Jordan-Wigner transformation,  
Eq.~\eqref{fermion2qubit} maps some local fermionic operators to non-local 
qubit operators. Consequently, local fermionic errors in Eq.~\eqref{errors} are 
mapped to non-local qubit errors. The fact that $X_1$ and $Z_1$ are highly 
non-local fermionic operators is key to understanding the utility of these 
codes.

It is straightforward to verify that the MZMs act non-trivially only on the 
`MZM' and `parity' qubits, and act trivially on the remaining $2n-2$ qubits. 
These $2n-2$ qubits indeed serve as the gauge degrees of freedom in the code 
$\mathcal{C}_0$ discussed in \S\ref{sec:bgenc}. All Pauli operators supported 
on the gauge qubits alone belong to the gauge group. The only stabilizer of 
$\mathcal{C}_0$, which is the MZM-parity operator $(i\gamma_1\gamma_2)
(i\gamma_3\gamma_4)$, corresponds to $Z_{\rm parity}$ under the map $\nu$. In 
other words, the projector on the stabilizer space projects onto the 
$\ket{0}_{\rm parity}$ state of the `parity' qubit. When the `parity' qubit is 
in the $\ket{0}_{\rm parity}$ state, the state of the `MZM' qubit is the 
logical state of $\mathcal{C}_0$. That is, $\ket{\bar{\zeta}} 
= \ket{\zeta}_{\rm MZM} \ket{0}_{\rm parity} \ket{\psi}_{\{(p,l)\}}$ for 
$\zeta = 0,1$, where $\ket{\psi}$ is an arbitrary state of the gauge qubits.  

In contrast, for the code $\mathcal{C}_{\rm GS}$ discussed in 
\S\ref{sec:bgec}, the logical states are $\ket{\bar{\zeta}} 
= \ket{\zeta}_{\rm MZM} \ket{0}_{\rm parity}\otimes_{(p,l)}\ket{0}_{(p,l)}$. 
These states in fact span the stabilizer space, and there are no gauge qubits. 
Observe that in both the codes $\mathcal{C}_0$ and $\mathcal{C}_{\rm GS}$, 
the logical qubit is encoded in the subsystem invariant under $\mathcal{G}$ 
of the subspace stabilized by $\mathcal{S}$. This in fact follows directly
from the definition of subsystem codes~\cite{poulin05}.

Consider an error-detecting code $\mathcal{C}$ in which a subset 
\begin{equation}
    S \subseteq \{\text{`parity'},(p,l),\ p=1,2,\ l=1,\dots,n-1\}
\end{equation} 
of qubits act as stabilizer qubits and the complementary subset of qubits 
\begin{equation}
    G = \{\text{`parity'},(p,l),\ p=1,2,\ l=1,\dots,n-1\}\setminus S
\end{equation}
act as the gauge qubits. Thus, the codes we consider encode
one logical qubit, the state of which is identical to the
state of the `MZM' qubit when the stabilizer qubits labeled by $\{r \in S\}$
are all in $\ket{0}_r$ state. If any of the stabilizer qubits is
in $\ket{1}_r$ state, then the state is outside the code space. 
The stabilizer group for such a code is 
\begin{equation}
    \mathcal{S} = \langle \{Z_r,\ r \in S\}\rangle,
\end{equation}
where $\langle\{\bullet\}\rangle$ denotes the group generated by $\{\bullet\}$. 
The gauge group $\mathcal{G}$ is generated by all Pauli operators acting on the qubits in $G$~\cite{dauphinais23} together with the stabilizers, that is
\begin{equation}
    \mathcal{G} = \langle\{Z_r,\ r \in S\}\cup\{i\}\cup\{X_r,Z_r,\ r \in G\}\rangle.
\end{equation}
In fact, as is well known in general, $\mathcal{G}$ alone uniquely defines the code $\mathcal{C}$.
For this code , $l_{\rm even}(\mathcal{C})$ is defined to be the length of the 
smallest weight Majorana fermionic operator which implements a non-trivial 
logical operation. Any error of weight less than $l_{\rm even}$ is detectable 
given access to measurements of all stabilizers in Eq.~\eqref{stab} 
individually. 

In this sense, the set $S$ of stabilizer qubits uniquely defines an 
error-detecting code. For a positive integer $d < n/2$, the code 
$\mathcal{C}_d$ consists of the stabilizer qubits  
\begin{multline}
    S_{d} = \text{`parity'} \cup \{(p,l),\ p=1,2,\ l=1,\dots,d\} \\ 
    \cup \{(p,l),\ p=1,2,\ l=n-d,\dots,n-1\}.
\end{multline}
We denote the corresponding stabilizer group by $\mathcal{S}_d$ and the 
corresponding gauge group by $\mathcal{G}_d$. Therefore, $\mathcal{C}_d$ is the
code uniquely defined by its gauge group $\mathcal{G}_d$.

\begin{figure}[t]
    \centering
    \includegraphics[width = \columnwidth]{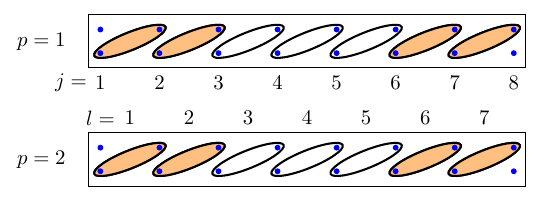}
    \caption{The support of stabilizers (shaded) of $\mathcal{C}_2$ on a tetron of length $n=8$.}
    \label{fig:CodeDistance}
\end{figure}

The image of the gauge group $\mathcal{G}_d$ under the mapping $\nu$ in 
Eq.~\eqref{fermion2qubit}, namely $\nu(\mathcal{G}_d)$, yields a Majorana 
fermionic code $\nu(\mathcal{C}_d)$. For brevity, we drop $\nu$ from the 
notation henceforth;  whether an operator is in the qubit Hilbert space or the 
tetron Fock space will be clear from the context. The following lemma provides 
the even-parity distance $l_{\rm even}(\mathcal{C}_d)$.
\begin{lemma}
The even-parity distance of $\mathcal{C}_d$ is
\label{lem:leven}
    \begin{equation}
        l_{\rm even}(\mathcal{C}_d) = 4d+4.
    \end{equation}
\end{lemma}
\begin{proof}
The proof is obvious given the locations of the stabilizers on the Kitaev 
tetron (see Fig.~\ref{fig:CodeDistance}). 
Any operator that implements a logical operator must contain terms 
quadratic in MZMs. An operator with even total-parity that commutes with all 
stabilizers and includes two MZMs must have weight $\ge 4d+4$.
\end{proof}

\subsection{Error detectability with respect to quasiparticle detection}
The discussion of error-correcting and error-detecting codes in the literature 
assumes that noisy measurements and unitary operations on all physical degrees 
of freedom are accessible. Since the resources available for stabilizer 
measurements in the system under consideration are physically constrained, 
it is important to 
clarify which error-detecting codes are compatible with these constraints. 
QP detection with spatial resolution $\lambda \le 1$ allows measurement of 
individual stabilizers, which enables the detection of any error with weight 
less than $l_{\rm even}$. This is not the case if only QP detections with low 
spatial resolutions are accessible, however. It is important, therefore, to
identify which codes remain detectable for a given $\lambda$, as all errors 
with weight less than $l_{\rm even}$ can be detected by a QP detector with 
spatial resolution $\lambda$. 
\begin{definition}
\label{def:lambdadet}
The code $\mathcal{C}_d$ is error-detecting with respect to $\lambda$ if
there exists a set of disjoint intervals $\mathcal{I} = \{I_{p,j} 
\subseteq \{1,\dots,n\}_p,\ |I_{p,j}| \ge \lambda\}$ such that 
\begin{equation}
\label{edcondition1}
    P_{\rm parity}P_{\mathcal{I}} = P_{\mathcal{S}},
\end{equation}
where 
\begin{equation}
\label{edprojector}
    P_{\mathcal{I}} = \prod_{p,j}(\mathds{1}-\theta(\hat{N}_{\rm qp}^{I_{p,j}}))
\end{equation}
is the projector on the space with no detectable QP excitations,
\begin{equation}
    P_{\mathcal{S}} = \prod_{r\in S_d} \left(\frac{\mathds{1}+Q_r}{2}\right) 
\end{equation}
is the projector on the stabilizer space,
and 
\begin{equation}
    P_{\rm parity} = \frac{(\mathds{1}+Q_{\rm parity})}{2}
\end{equation}
is the projector on the even MZM-parity space.
\end{definition}
This definition implies that any error that takes the system out of the
stabilizer space is detectable. More concretely, for any even total-parity 
error $E$ with ${\rm wt}(E)<4d+4$,
\begin{equation}
\label{edcondition2}
    \tr_{G}(P_{\rm parity}P_{\mathcal{I}}EP_{\mathcal{S}}E^\dagger 
		P_{\mathcal{I}}P_{\rm parity} ) \propto \tr_{G}(P_{\mathcal{S}}).
\end{equation}
In other words, if the outcomes of all binary measurements
$\{\mathcal{M}_{I_{p,j}}\}$ are $0$, then the post-measurement state is equal 
to the code state before being subjected to the error operator $E$. 
Consequently, if an even total-parity error $E$ is not detectable and acts 
non-trivially on the logical qubit, then ${\rm wt}(E)\ge 4d+4$.

\begin{theorem}
\label{thm:lambdaed}
For any $d<n/2$, the code $\mathcal{C}_{d}$  
is error-detecting with respect to $\lambda$ for $\lambda \le d$.
\end{theorem}
\begin{proof} 
It suffices to show that $P_{\mathcal{I}}$ defined in 
Eq.~\eqref{edprojector} with respect to the intervals
\begin{equation}
\label{eq:edintervals}
    I_{p,1} = \{1,\dots,d\}_p, \ I_{p,2} = \{n-d,\dots,n-1\}_p, \quad p=1,2
\end{equation}
satisfies Eq.~\eqref{edcondition1}. By Eq.~\eqref{fermion2qubit}, one obtains
\begin{eqnarray}
		\mathds{1}-\theta(\hat{N}_{\rm qp}^{I_{p,j}})&\xrightarrow{\nu^{-1}}&
		\sum_{l \in I_{p,j}} (\mathds{1}+Z_{p,l})/2\nonumber \\
    &\xrightarrow{\nu^{-1}}& \bigotimes_{l \in I_{p,j}} \ket{0}_{p,l}\bra{0},
\end{eqnarray}
and therefore
\begin{equation}
    P_{\mathcal{I}} \xrightarrow{\nu^{-1}} \bigotimes_{p=1,2,\ l \in 
		\cup_{j} I_{p,j}} \ket{0}_{p,l}\bra{0}.
\end{equation}
Given that $\cup_{p,j} I_{p,j}\cup\{\text{`parity'}\} = S_d$,
\begin{equation}
    P_{\rm parity}P_{\mathcal{I}} \xrightarrow{\nu^{-1}} \bigotimes_{r\in S_d} 
		\ket{0}_r\bra{0} \xrightarrow{\nu} P_{\mathcal{S}}.
\end{equation}
Therefore, $P_{\rm parity}P_{\mathcal{I}}=P_{\mathcal{S}}$,
which completes the proof. 
\end{proof}

\noindent Theorem~\ref{thm:lambdaed} implies that given a QP detector with 
inverse resolution $\lambda$, any error-detecting code $\mathcal{C}_d$ with 
$d\ge \lambda$ is implementable in such a way that any even total-parity error 
$E$ that causes an undetected non-trivial logical operation has 
${\rm wt}(E)\ge 4d+4$.

It is worthwhile to note that given a detector with resolution $\lambda=1$, the 
code $\mathcal{C}_0$ is not only error-detecting but also error-correcting. In 
this limit, QP detectors can implement all stabilizer measurements 
individually, and so the situation is analogous to the Kitaev tetron code. Any 
even total-parity error $E$ with ${\rm wt}(E)\le n-1$ can be corrected by an
appropriate (possibly linear) unitary operation on the MZMs. No operation on 
the bulk Majorana operators is required for the recovery operation. 
Unfortunately, this extraordinary capability cannot be achieved with QP 
detectors of inverse resolution $\lambda \ge 2$. Realizing this scenario 
requires a very high level of control over microscopic degrees of freedom, and 
therefore the investigation of this scenario, although very promising, is 
outside the scope of the present work.

\subsection{Extension of error-detecting codes to parameters away from the fixed point}
The observables $\hat{N}_{\rm qp}^{I}$ depend implicitly on the parameters 
$(\mu,w,\Delta)$ of the system, Thus, if the Kitaev tetron is tuned away from 
the fixed point $\mu=0$ and $w=\Delta=1$, the stabilizers defined in 
Eq.~\eqref{stab} are no longer directly accessible. It is nevertheless possible 
to construct error-detecting codes analogous to $\mathcal{C}_d$ for arbitrary 
parameter values throughout the topological phase $|\mu|<|w|$.

The key observation is that the WQP operators satisfy the same commutation and 
anticommutation relations for any parameter values in the topological regime.
Let $\{{\phi}_j\}$ and $\{{\phi}_{j'}'\}$ be the sets of WQP annihilators at 
two different parameter values $(\mu, w, \Delta)$ and $(\mu', w', \Delta')$ 
respectively. These satisfy the same canonical anticommutation rules:
\begin{equation}
    \{{\phi}_j, {\phi}_{j'}^{\dag}\} 
		= \{{\phi}_{j}', \left({\phi}'_{j'}\right)^{\dag}\} 
		= \delta_{jj'}\mathds{1}.
\end{equation}
Therefore, the mapping in Eq.~\eqref{fermion2qubit} 
is valid for all parameter values 
in the topological regime.
New codes $\mathcal{C}_d(\mu,w,\Delta)$ can then be 
constructed in exactly the same way as the $\mathcal{C}_d$ were constructed at 
the fixed point.

The newly constructed codes $\mathcal{C}_d(\mu,w,\Delta)$ have stabilizer 
operators of the form 
\begin{equation}
    Q_{p,l}(\mu,w,\Delta) = (-1)^{\left({\phi}_{p,l}'\right)^\dagger{\phi}_{p,l}'},
\end{equation}
which are exponentially localized in space, but not as compactly as at the 
fixed point. Lemma~\ref{lem:leven}, Definition~\ref{def:lambdadet} and 
Theorem~\ref{thm:lambdaed} follow as long as the weight ${\rm wt}'$ of 
the errors are defined in terms of WQPs. Consequently, the codes 
$\mathcal{C}_d(\mu,w,\Delta)$ are error-detecting codes for errors that are 
local in the WQP basis. Note that the errors that are local in the original 
basis $\{a^\dagger_j,a_j\}$ are not necessarily local in the WQP basis; 
however, due to exponential localization of WQPs, the overlap of a 
low-${\rm wt}'$ error on a high-${\rm wt}$ error is exponentially 
small. For example, the overlap of an error with constant ${\rm wt}'$ 
on an error $E$ with ${\rm wt}(E) \ge 4d+4$ is exponentially small in $d$. 
This implies that an error with constant ${\rm wt}'$ can be detected 
with probability $1-p_{\rm fail}$, where $p_{\rm fail}$ is exponentially small 
in $d$. One may therefore think of $\mathcal{C}_d(\mu,w,\Delta)$ as an 
approximate error-detecting code with respect to low-${\rm wt}$ errors.

\section{Quantifying the advantage offered by quasiparticle detection}
\label{sec:performance}

In this section, we compare the performance of qubits encoded in a Kitaev 
tetron with and without QP detection, in the presence of QP 
poisoning events. Having shown that the code $\mathcal{C}_d$ can be implemented 
by using a detector with $\lambda \le d$, it suffices to compare the
performance of the codes $\{\mathcal{C}_d\}$ against each other. The metric 
that we use for this comparison is the bit-flip error rate $p_{\rm bit-flip}$ 
and the erasure rate $p_{\rm erasure}$ for the logical qubit encoded in 
$\mathcal{C}_d$. 

\subsection{The noise model}
Assume a noise model where the system undergoes the channel
\begin{equation}
    \rho \mapsto \mathcal{E}(\rho) = \mathcal{E}_{2,n} \circ \dots 
		\mathcal{E}_{2,1} \circ \mathcal{E}_{1,n} \circ \dots \circ 
		\mathcal{E}_{1,1}(\rho),
\label{eq:erroroperator}
\end{equation}
where 
\begin{equation} 
    \mathcal{E}_{p,j}(\rho) = q(-ic_{2j-1}c_{2j})\rho (-ic_{2j-1}c_{2j}) 
		+ (1-q)\mathds{1}
\end{equation}
for some $q \in [0,1]$. The error channel $\mathcal{E}_{p,j}$ implements the 
elementary error operator $E_{p,j}$ -- c.f.\ Eq.~\eqref{errors} -- with 
probability $q$, and all elementary errors are independent. Therefore, for 
every $\mathcal{J}_1 , \mathcal{J}_2 \subseteq \{1,\dots,n\}$, the error 
channel in Eq.~\eqref{eq:erroroperator} implements the error operator 
\begin{equation}
\label{eq:erroroperator2}
    E(\mathcal{J}_1 , \mathcal{J}_2)  = \prod_{p=1,2}\prod_{j \in \mathcal{J}_{p}} E_{p,j}     
\end{equation}
with probability 
\begin{equation}
    \text{prob}(\mathcal{J}_1 , \mathcal{J}_2) =   q^{|\mathcal{J}_1 + \mathcal{J}_2|}(1-q)^{2n-|\mathcal{J}_1 + \mathcal{J}_2|}.
\end{equation}
It is straightforward to verify that 
\begin{equation}
\sum_{\mathcal{J}_1\subseteq \{1,\dots,n\}} \sum_{\mathcal{J}_2
\subseteq \{1,\dots,n\}} \text{prob}(\mathcal{J}_1 , \mathcal{J}_2)=1,
\end{equation}
so that the channel in 
Eq.~\eqref{eq:erroroperator} is both trace-preserving and completely positive 
owing to the Kraus form. At the fixed point, $E_{p,j}$ for $j \ne 1,n$ excites 
or de-excites two QPs with probability $q$. Therefore, $q$ is roughly twice the 
rate of QP excitations per unit time per unit length, and is assumed to be 
independent of $n$. For $j \in \{1,n\}$, the error operator $E_{p,j}$ excites 
or de-excites, with probability $q$, one QP above bulk energy gap while also 
acting on the adjacent MZM, which leads to QP poisoning. Consequently, $q$ can 
also be interpreted as the rate of QP poisoning per MZM. The rate of relaxation 
of excited QPs is the same as the rate of excitation in this model, so this 
model mimics the noise due to coupling to a thermal bath at a very high 
temperature.

A couple of attributes of the noise model deserve a justification here.
First, the noise model does not account for excited QPs spreading
due to dispersion. Such an effect is present in any model 
of a topological superconductor away from the fixed point. 
The noise model effectively assumes that the rate of QP detection is high
compared to the time required for QPs to escape the region of detection.
This assumption keeps the analysis simple and provides a proof-of-principle 
demonstration of the erasure conversion scheme.

Second, the noise modelled by Eq.~\eqref{eq:erroroperator} preserves total fermion parity.
In fact, the erasure conversion scheme is incapable of detecting local errors 
that are linear combinations of MZMs. The latter kind of errors are caused by extrinsic QP 
poisoning processes, in which a fermion is injected into the system from its environment. 
Nevertheless, such errors can often be suppressed by tuning some hardware parameters.
For example, in nanowire-based implementations, a high charging energy for the 
superconducting island prohibits isolated electrons from hopping onto the
nanowire from the neighboring quantum dots~\cite{activecorrection4}. The residual linear errors can be
corrected by the error-correcting code in the layer above the error-detecting code.
It is worth a remark that while the errors considered here preserve the total
parity, they do not necessarily preserve the Majorana parity $\expval{Q_{\rm parity}}$. 
Consider, for instance, the error 
$c_{1,1} c_{1,2} = \gamma_1 c_{1,2}$ at the fixed point. 
This error is evidently bilinear and therefore
has even total parity. However, in a low energy model, this same error is 
equivalently described by the linear operator $\gamma_1$,
as all bulk degrees of freedom are traced out, and therefore has odd Majorana parity.
This is the reason why it is crucial for most other works on Majorana codes 
to consider errors of both odd and even (Majorana) parity~\cite{Vijay2015, mclauchlan22},
whereas it suffices 
for this work to only consider errors of even total parity.

\begin{figure}[t]
    \centering
    \includegraphics[width = \columnwidth]{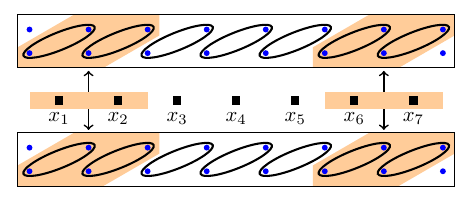}
    \caption{The region under detection for implementing $\mathcal{C}_2$ using 
		a detector with inverse resolution $\lambda = 2$ on a tetron of length 
		$n=8$.}
    \label{fig:Detector}
\end{figure}

\subsection{Error rates at the fixed point}

The calculation of error rates at the fixed point is tractable because the 
error operators act locally on stabilizers, with which they either commute or 
anticommute. Suppose the qubit is initialized in the 
$\ket{\bar{0}} \in \mathcal{F}_{2n}$ state of the codespace 
$\mathcal{C}_d$, where $d \in \{0,\dots,\lfloor n/2\rfloor -1\}$. Consider 
first the the rate of erasure $p_{\rm erasure}$. Recall that the logical qubit 
will be erased and resetduring decoding 
if either the QP detection or the MZM parity measurement 
signals an error. It is simpler to calculate the rate $1-p_{\rm erasure}$, which 
equals the probability of neither the QP detectors nor the MZM-parity 
measurement signaling an error. Suppose $E(\mathcal{J}_1 , \mathcal{J}_2)$
includes some Majorana operators and excludes at least one Majorana operator
other than the MZM in a single detection region, say $I_{1,1}$ (see 
Fig.~\ref{fig:Detector}). Then there exists at least one stabilizer $Q_{p,l}$ 
with support on one excluded and one included Majorana operator, and therefore 
$\expval{Q_{p,l}} = -1$. Recall that a flipped stabilizer is equivalent to the
presence of a QP, and therefore the QP detector monitoring the $I_{1,1}$ region 
would click and subsequently the qubit will be erased. If the only excluded 
Majorana operator is the MZM in that detection region, then no $Q_{p,l}$ would 
click. However, if $E(\mathcal{J}_1 , \mathcal{J}_2)$ includes an odd number of 
MZMs, then the MZM-parity measurement signals an error, and qubit will be erased 
again. One therefore concludes that the qubit is not erased only if the 
error operator $E(\mathcal{J}_1 , \mathcal{J}_2)$ includes all Majorana 
operators in zero, two or all four detection regions in 
$\mathcal{I} = \{I_{1,1},I_{1,2},I_{2,1},I_{2,2}\}$.

\begin{figure*}[t]
    \includegraphics[width = 2.1\columnwidth]{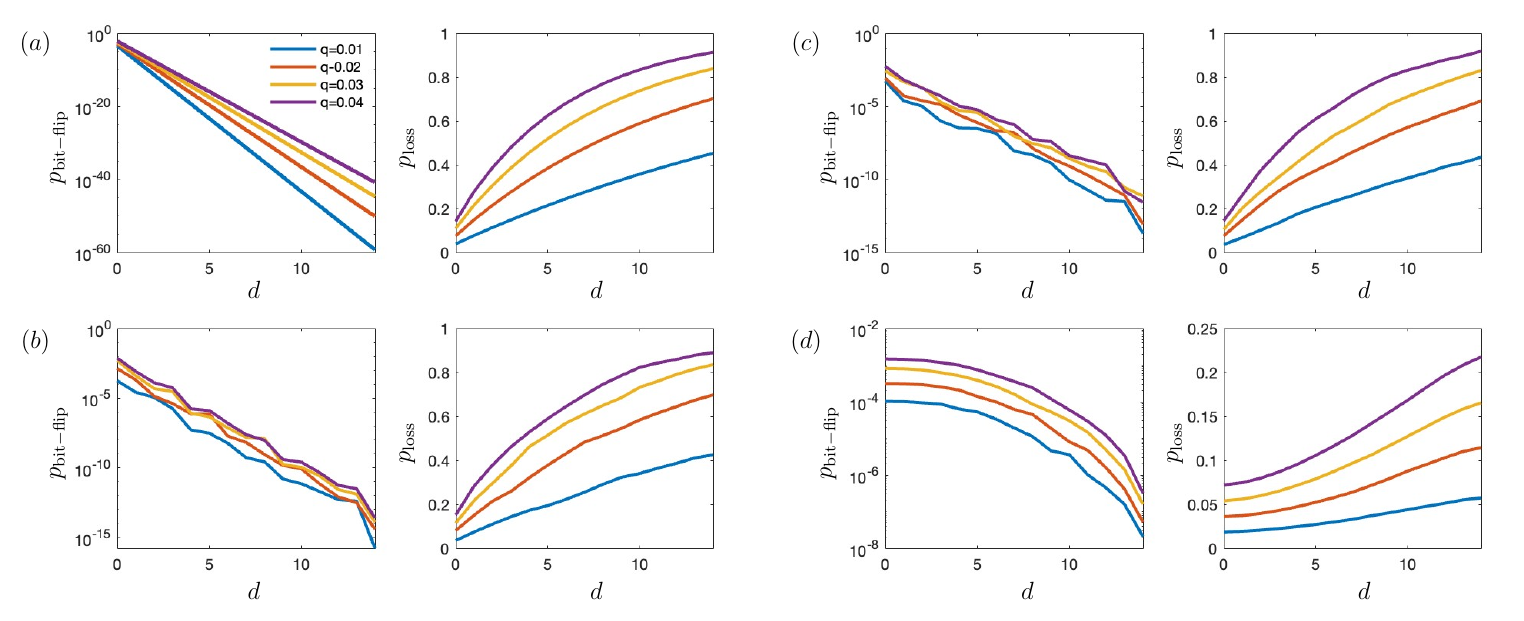}
    \caption{Bit-flip rate (left column) and erasure rate (right column) as 
		a function of $d$ for various values of noise strength $q$ at (a) the fixed 
		point, i.e. $\mu = 0$, $w = 1$, $\Delta = 1$, (b) away from the fixed point 
		at $\mu = 0.3$, $w = 1$, $\Delta = 0.4$, $n = 30$, (c) with uniform on-site 
		disorder of strength $\delta\mu = 0.3$ at the same parameter values, and (d) near 
		the topological phase boundary at $\mu = 0.99$, $w = 1$, $\Delta = 0.1$, 
		$n=30$. 
    \label{fig:errorrates}}
\end{figure*}

The probability of $E(\mathcal{J}_1 , \mathcal{J}_2)$ including no Majorana 
operators in any of the four detection regions is $(1-q)^{4d+4}$. This is 
because each detection region has overlap with $d+1$ elementary errors, and the 
probability of each elementary error not occurring is $1-q$. Similarly, the 
probability of $E(\mathcal{J}_1 , \mathcal{J}_2)$, including all Majorana 
operators in all four detection regions, is $q^{4d+4}$. Finally, the 
probability of $E(\mathcal{J}_1 , \mathcal{J}_2)$ including all Majorana 
operators in two out of four detection regions is $6q^{2d+2}(1-q)^{2d+2}$, 
where the factor of $6$ comes from the six ways of choosing two detection 
regions out of four. Adding all these probabilities, one obtains
\begin{equation}
    1- p_{\rm erasure} = (1-q)^{4d+4} + 6q^{2d+2}(1-q)^{2d+2} + q^{4d+4}.
	\label{eq:ploss}
\end{equation}

Consider second the bit-flip error rate $p_{\rm bit-flip}$. This is defined to 
be the probability of a qubit initialized in the state $\ket{\bar{0}}$ 
undergoing a bit flip conditional on it not being erased. A bit-flip error occurs 
if the error operator $E(\mathcal{J}_1 , \mathcal{J}_2)$ has support on exactly 
one MZM in each of the sets $\{\gamma_1,\gamma_2\}$ and 
$\{\gamma_3,\gamma_4\}$. Such a bit-flip error goes undetected if 
$E(\mathcal{J}_1 , \mathcal{J}_2)$ includes all Majorana operators in an even 
number of detection regions, as explained in the previous paragraph. Both these 
conditions are satisfied if $E(\mathcal{J}_1 , \mathcal{J}_2)$ includes all 
Majorana operators in one region on $p=1$ chain and one region on $p=2$ chain.
Therefore, the logical bit-flip error probability conditional on detection of 
errors, which is the logical bit-flip error rate, is
\begin{equation}
    p_{\rm bit-flip} = \frac{4q^{2d+2}(1-q)^{2d+2}}
		{(1-q)^{4d+4} + 6q^{2d+2}(1-q)^{2d+2} + q^{4d+4}}.
	\label{eq:pbitflip}
\end{equation}
Here the denominator is simply the probability of the qubit not being lost,
which is $1-p_{\rm erasure}$ calculated in Eq.~\eqref{eq:ploss}.
Note that both the erasure rate as well as the 
bit-flip error rate are independent of $n$ at the fixed point. 

The error rates at the fixed point, Eqs.~(\ref{eq:ploss}) and 
(\ref{eq:pbitflip}), are plotted in Fig.~\ref{fig:errorrates}(a). 
Focus on the parameter regime $q \ll 1/d \ll 1$, which requires QP detectors 
with spatial resolution higher than the density of QPs excited by noisy 
operations. In this regime, the erasure probability scales as 
$p_{\rm erasure}\sim(4d+4)q$, i.e.\ increases linearly in the code distance, while
the bit-flip probability scales as $p_{\rm bit-flip}\sim 4q^{2d+2}$, i.e.\ 
decreases exponentially with the code distance. QP detection therefore allows
trading off a significant fraction of Pauli errors for erasure errors.
If $d$ is chosen too high, then $p_{\rm erasure} \approx 1$. A critical value of 
$d$ is likely to be the most advantageous, and will depend on the 
error-correction scheme in use. Assuming access to arbitrarily high resolution, 
suppose the error-correction scheme tolerates a constant $p_{\rm erasure}$. In 
terms of the erasure rate, one can express the bit-flip error rate as
\begin{equation}
\label{tradeoff}
    p_{\rm bit-flip} \approx 4q^{(p_{\rm erasure}/4q)}.
\end{equation}
Therefore, $p_{\rm bit-flip}$ is suppressed exponentially in $1/q$ for fixed 
$p_{\rm erasure}$. Equation~\eqref{tradeoff} shows a glimpse of why QP detection 
could be extremely powerful for achieving fault tolerance and for reducing 
resource overhead imposed by error correction.

\subsection{Error rates away from the fixed point}

\subsubsection{Formalism}

If the Kitaev tetron is tuned to parameters away from the fixed point, then 
the error rates are not tractable analytically. Results are instead obtained 
numerically using the covariance matrix framework, and the basis defined by 
Wannier QPs makes the calculation considerably simpler. Define Majorana 
operators $\{c'_j,\ j=1,\dots 2n\}$ using the relations
\begin{eqnarray}
    &&{\phi}^\dagger_{p,l} = (c'_{p,2l} + ic'_{p,2l+1})/2,\quad 
    {\phi}_{p,l} = (c'_{p,2l} - ic'_{p,2l+1})/2, \nonumber \\
    &&\gamma_1 = c'_{1,1},\quad \gamma_2 = c'_{1,2n}, \quad 
    \gamma_3 = c'_{2,1},\quad \gamma_4 = c'_{2,2n}.
\end{eqnarray}
The covariance matrix $M_0$ of $\ket{\bar{0}}$ is defined with respect to 
$\{c'_{p,j},\ p=1,2,\ j=1,\dots 2n\}$, c.f.\ Lemma~\ref{lem:M0}.
The erasure rate after the state has been subjected to $\mathcal{E}$ is 
given by
\begin{equation}
    p_{\rm erasure} = 1-\expval{P_{\mathcal{S}}}_{\mathcal{E}(\rho)} = 1-\expval{\prod_{r \in S_d}\frac{(\mathds{1}+Q_r)}{2}}_{\mathcal{E}(\rho)}.
\end{equation}
The expectation value is obtained by first sampling an error operator $E$,
which is performed by independently sampling two strings  in $\{0,1\}^{n}$ from 
the probability distribution $\{q,1-q\}^{\times n}$. Let $\mathcal{J}_1, 
\mathcal{J}_2$ be the sets of positions of $1$s in the two strings 
respectively. The corresponding error operator 
$E(\mathcal{J}_1, \mathcal{J}_2)$ is given by Eq.~\eqref{eq:erroroperator2}. 
The normalized state after the action of $E$ is 
$E\ket{\bar{0}}/\norm{E\ket{\bar{0}}}$. To calculate the covariance matrix of 
$E\ket{\bar{0}}$, $E$ is expressed as a unitary operator generated by a 
quadratic Hamiltonian, as follows.
\begin{lemma} The following identity holds:
\begin{equation}
    \prod_{p=1,2}\prod_{j\in \mathcal{J}_p}E_{p,j} = 
    i^{|\mathcal{J}_1+\mathcal{J}_2|}\exp\left(-\frac{i\pi}{2}\sum_{p=1,2}\sum_{j\in \mathcal{J}_p} E_{p,j}
		\right).
\end{equation}
\end{lemma}
With this lemma, one can write $M_{E\ket{\bar{0}}} = R_sM_0R_s^{\rm T}$, where
\begin{equation}
    R_s = \exp\left[\frac{\pi}{2}\sum_{j\in s}\left(\ket{2j-1}\bra{2j}
		-\text{H.c.}\right)\right].
\end{equation}
The matrix elements of the effective Hamiltonian 
$i(\sum_{j\in s}\ket{2j-1}\bra{2j}-\text{H.c.})$ are computed in the WQP basis 
in order to obtain the matrix representation of $R_s$, which in turn generates
$M_{E\ket{\bar{0}}}$. Finally, to calculate $p_{\rm erasure}$ corresponding to 
$E$, the following lemmas are used: 
\begin{lemma}
\label{lem:projectorwick1}
The projector on the $+1$ eigenspace of 
$Q_{p,l}$ is
    \begin{equation}
        P_{p,l} = \frac{(\mathds{1}+Q_{p,l})}{2} 
				= \frac{c_{p,2l}(c_{p,2l}-ic_{p,2l+1})}{2}.
    \end{equation}
\end{lemma}

\begin{lemma}
\label{lem:projectorwick2}
For $P_{\rm parity} = (\mathds{1}+(-i\gamma_1\gamma_2)(-i\gamma_3\gamma_4))/2$, 
\begin{equation}
    P_{\rm parity} (-i\gamma_1\gamma_2)P_{\rm parity} = -i\gamma_1\gamma_2 
		- i\gamma_3\gamma_4.
\end{equation}
\end{lemma}

The bit-flip error rate is defined to be
\begin{equation}
    p_{\rm bit-flip} = (1-\expval{\bar{Z}})/2,\quad \expval{\bar{Z}} = 
    \frac{\expval{P_{\mathcal{S}}(-i\gamma_1\gamma_2)
		P_{\mathcal{S}}}_{\mathcal{E}(\rho)}}
		{\expval{P_{\mathcal{S}}}_{\mathcal{E}(\rho)}},
\end{equation}
which can be calculated using similar techniques. The numerator of $\bar{Z}$ is 
computed using Lemmas~\ref{lem:projectorwick1} and \ref{lem:projectorwick2}.

\subsubsection{Results}

For parameters away from the fixed point, noise can induce errors due to the
overlap of MZMs. Also, now elementary errors have non-zero overlap on spatially 
distant stabilizers, which results in a larger fraction of high-weight errors, 
and consequently a larger fraction of undetected logical errors. The erasure 
and bit-flip error rates as a function of $q$ are plotted in 
Fig.~\ref{fig:errorrates}(b) for parameters away from the fixed point but 
nevertheless away from the topological phase boundary. The qualitative 
behaviors are similar to those at the fixed point and one observes a clear 
trade-off between Pauli and erasure errors. Consequently, a critical value 
of $d$ is likely to be the most advantageous, as was inferred from the 
analysis at the fixed point. The qualitative behavior of error rates also 
remains the same in the presence of weak disorder, 
Fig.~\ref{fig:errorrates}(c). However, the absolute rate of bit-flip errors in 
both Fig.~\ref{fig:errorrates}(b) and Fig.~\ref{fig:errorrates}(c) is higher 
than at the fixed point, as one would expect due to higher overlap of Majorana 
modes with distant WQPs. Such a comparison of absolute error rates across 
parameter values is somewhat unfair, though, as the noise strength $q$ is 
assumed to be independent of the parameter values. In real systems, the noise 
strength is likely to be higher for systems with parameters that lead to a low 
bulk energy gap. 

The behaviors of the erasure and bit-flip errors change as the parameters 
approach the topological phase boundary, Fig.~\ref{fig:errorrates}(d), because 
of the increased spatial extent of both the MZMs and the WQPs (for the latter 
refer to Fig.~\ref{fig:wannier}). Consequently, a local error can induce 
coherence between a MZM and a distant WQP. Nevertheless, the Pauli error rates 
continue to be exponentially suppressed as long as $d$ is much larger than the 
characteristic length of the MZMs and the WQPs, both of which 
increase as the many-body gap decreases. This constrains the domain of $q$
over which Pauli error decrease exponentially, because this 
behavior is expected to hold only for $q \ll 1/d$. For a constant QP poisoning 
rate $q$, the effective erasure rate is smaller in magnitude at fixed $d$ as the 
parameters approach the topological phase boundary. The origin of this effect 
can be traced to fermionic interference effects arising from the fact that the 
QP basis is nearly conjugate to the basis in which errors are local.

It is worth emphasizing that the exponential suppression of Pauli error rates 
is a direct consequence of the exponential localization of WQPs. Consider the 
error operator $E_{1,1}E_{2,1} = -c_{1,1}c_{1,2}c_{2,1}c_{2,2}$. At the fixed 
point, this error is not detected by the stabilizer $Q_{\rm parity}$ because
it involves two MZMs, namely $\gamma_1 = c_{1,1}$ and $\gamma_3 = c_{2,1}$, but
it is detected by QP detectors as long as $d \ge 1$. If the 
same error occurs on a tetron with parameters away from the fixed point, it may 
go undetected by both $Q_{\rm parity}$ and the QP detectors with some non-zero
probability $p_{\rm fail}$.  The probability $p_{\rm fail}$ is proportional to the 
overlap of the error with the WQP operators supported in the region not under
detection. Due to exponential localization of WQPs,
$p_{\rm fail}$ is exponentially small in $d$, which in turn accounts for the
exponential suppression of the Pauli error rate.

\begin{figure} 
    \includegraphics[width = 0.9\columnwidth]
		{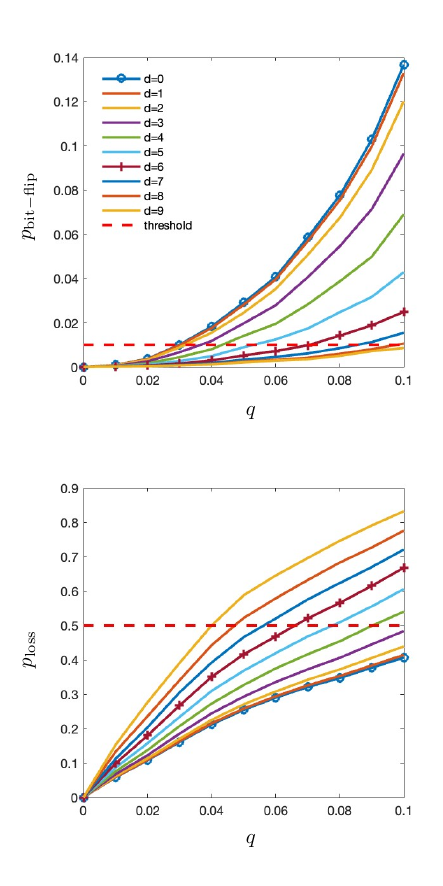}\\
    \caption{Bit-flip error rate (top) and erasure rate 
            (bottom) as a function of the noise strength $q$ 
            for various values of $d$ 
		at $\mu = 0.95$, $w = 1$, $\Delta = 0.05$ and $n = 30$. 
            The blue curves with `$\circ$' markers correspond to 
            error rates with no QP detection ($d=0$), and the 
            maroon curves with `$+$' markers correspond to 
            $d=6$. The red dashed curves indicate
            typical threshold values for the respective error rates.
    \label{fig:errorrates2}} 
\end{figure}

The error rates for logical qubits discussed above are only meaningful if it is 
assumed that $\mathcal{C}_d$ constitutes the first layer of a concatenated 
error-correcting code. In that case, the qubits encoded in $\mathcal{C}_d$ in 
individual Kitaev 
tetrons are considered as the physical qubits for the next layer of error 
correction, and not as the logical qubits per se. Suppose that QP detection is 
performed in parallel with MZM parity measurement on each tetron, and both 
these measurements are perfect. If either of the QP or the MZM parity detectors
signal an error, then the state is unrecoverable; the qubit is then discarded 
and subsequently reset during the decoding protocol of the 
concatenated code. In this setting, the error rates for the logical qubit 
encoded in $\mathcal{C}_d$ would be treated as physical error rates for the 
next layer of the concatenated code. The threshold rates for erasure and Pauli 
errors respectively could then be calculated; these would depend, in general, 
on the details of the upper layers of the concatenated code. The implementation 
of such a concatenated code for any value of $d$ requires the same number of 
tetrons and the same amount of time per cycle, allowing for a fair comparison.

Figure~\ref{fig:errorrates2} depicts the erasure and Pauli error rates as a 
function of $q$ for various values of $d$ for parameters away from the fixed 
point and far from the topological phase boundary; the rates for $d=0$ 
correspond to no QP detection. The trade-off between Pauli and erasure rates
is clearly seen, as the bit-flip error rate decreases and the erasure rate 
increases with $d$. Assume that a second layer of quantum error-correcting code 
above $\mathcal{C}_d$ has a threshold rate of 1\% for Pauli errors and 50\% for 
erasure errors, values consistent with thresholds for topological 
codes~\cite{Stace2009,Stace2010}. For $d=0$, i.e. without QP detection, the 
bit-flip error rate is above the typical threshold for $q > q_{\rm th}^{(0)} 
\approx 0.03$. The erasure rate is well below the typical threshold for 
$q < q_{\rm th}^{(0)}$. Therefore, for the parameters chosen in 
Fig.~\ref{fig:errorrates2}, fault-tolerance can be achieved without QP 
detection for $q < q_{\rm th}^{(0)} \approx 0.03$. In comparison, if 
$\mathcal{C}_6$ is used for the first layer of the error-correcting code, then
fault-tolerance can be achieved for $q < q_{\rm th}^{(6)} \approx 0.06$. 
Since $q_{\rm th}^{(6)} > q_{\rm th}^{(0)}$, the concatenated error-correcting 
code can tolerate higher rate of QP poisoning errors with QP detection in the 
first layer. Note that a detector with spatial resolution $\lambda \le 6$ 
allows for the implementation of $d=6$, and therefore achieves the threshold 
value $q_{\rm th}^{(6)}$. For $d>6$, the threshold rate $q_{\rm th}^{(d)}$ is 
limited by the erasure rate and therefore drops below $q_{\rm th}^{(6)}$. 
Therefore, $d=6$ is optimal if $q_{\rm th}^{(d)}$ is the figure of merit and 
only requirements are $p_{\rm bit-flip}< 0.01$ and $p_{\rm erasure}<0.5$. 

Note that the precise values of $q_{\rm th}$ are not practically important as 
the noise model considered here is not tied to specific implementations. 
Likewise, a concatenated quantum error correcting code with a topological code 
in the outer layer may not be feasible or optimal for realistic tetron 
architectures. Considering its high threshold 
rate~\cite{Vijay2015, mclauchlan22}, the Majorana surface code modified to 
leverage QP detection measurements might be a better alternative than a 
concatenated code as considered above. Regardless, the numerical results above 
demonstrate that QP detection could be highly advantageous for achieving fault 
tolerance and for reducing the resource cost incurred by error correction.

\section{Conclusions}
\label{sec:conclusion}

The problem addressed in this work is to ascertain to what extent a spatially 
resolved detection of QPs can aid in improving the quality of topological 
qubits, encoded in MZMs. A model of QP detection with finite spatial resolution 
is proposed, via the construction of a suitable QP position operator. 
The eigenmodes of the QP position operator, 
which describe the WQPs, are proven to be exponentially localized. 
Majorana error-detecting 
codes are then designed that can be implemented on a tetron with the help of 
QP detection. These codes facilitate erasure conversion by
detecting local QP poisoning errors.
Finally, assuming that the dominant errors arise from 
intrinsic QP poisoning, we have shown that the resulting erasure conversion
scheme suppresses the Pauli error rates exponentially  
in the code distance, at the cost of a linearly increasing erasure rate.
This exponential suppression is tied to the exponential localization of WQPs.
While our analysis is done entirely on the Kitaev tetron model,
it can be easily extended to tetrons based on other models of 
topological superconductors, and so we expect an exponential suppression in Pauli  
error rate in those cases as well.
Thus, QP detection allows to remove a large fraction of Pauli errors at the
cost of increased erasure errors, which are easier to 
correct in a second layer of a concatenated quantum error-correcting code.
Therefore, QP detection is an appealing approach for 
tackling intrinsic QP poisoning errors in Majorana qubits, 
and thereby achieve fault tolerance.

Most important, QP detection bolsters the passive protection that is inherent
to topological systems such as those that support MZMs considered here. As
mentioned previously, the errors that are dominated by the overlap of the 
computational anyons, and are not mediated by QP poisoning, are suppressed 
exponentially in the distance between the computational anyons. Therefore, such 
errors can be suppressed by increasing the distance between anyons. However, 
the errors arising from QP poisoning are suppressed only in the bulk energy 
gap, which is not easy to tune in 
experiments~\cite{Goldstein2011,Budich2012,diego12,majoranaoperations,
Hu2015,majoranacode2,Ippoliti2016}. Our results show that QP detectors with 
spatial resolution can be effective in suppressing these QP poisoning errors.

Several questions remain to be answered in this context. 
One is how best to implement the position-resolving QP detector modeled in our work
in practical Majorana-based qubits. Because the QP detection region in our scheme 
overlaps with the MZMs, the main challenge is to detect QPs without destroying the information
stored in MZMs. A potential approach would be to leverage the energy separation between
the MZMs and the bulk QPs for this purpose. However, no experimental
measurement is known to our knowledge that can probe high-energy QPs 
without coupling to MZMs. Therefore, implementation of QP detection near the ends of the tetron 
as modeled in this work requires further research. Meanwhile, one may perform erasure
conversion sufficiently far from the MZMs so as to reduce disturbance to the MZMs. 
While such detectors are likely to be less effective for erasure conversion than the ones modeled
in this work, they may nevertheless lead to improved threshold QP poisoning rate.

Another question is how quickly one must perform QP detection 
measurements in practice, to ensure that the QP remains localized. A good 
estimate can be obtained by considering that the QPs must not
escape the detection region in time $\Delta t$ between two successive 
measurements.  To obtain a quantitative estimate, consider Majorana qubits
of semiconductor nanowires. These are proposed to be tens of micrometers in 
length~\cite{aghaee2022}, so assume a detection region of approximately 1~$\mu$m. The average 
QP speed in thermal equilibrium is given by 
$\expval{|v_{\rm qp}|}\sim v_F\sqrt{1/\beta\Delta}$ where $v_F$ is the Fermi
velocity, $\beta$ is the inverse temperature and $\Delta$ is the induced superconducting gap 
(see Appendix~\ref{app:qpvelocity}). Using typical 
values $\beta\Delta \sim 100$ and $v_F \sim 10^5$ m/s, one obtains
$\expval{|v_{\rm qp}|} \sim 10^4$~m/s and $\Delta t \sim 0.1$~ns. 
While challenging, these measurement timescales are currently 
achievable in semiconductor tunnel junction experiments~\cite{Liang2023},
and nanosecond timescales have been achieved in time-resolved 
transport measurements~\cite{Kamata2022}.
Furthermore, a significantly higher $\Delta t$ could suffice if the 
quasiparticles relax to the lower band edge due to inelastic scattering 
processes, if the chemical potential is set to a value with lower Fermi 
velocity $v_F$, and if a higher topological band gap could be achieved.

Another challenge is to devise optimal schemes for leveraging QP detection in
practical applications, including those beyond Majorana-based quantum 
computation. In Majorana-based architectures, the obvious next step is
to develop decoding strategies for the Majorana surface code and its variants
that leverage QP detection measurements. A strategy similar to the ones for 
tackling erasure in qubit surface code~\cite{Stace2009,Stace2010} is likely 
to succeed. For other topological platforms such as those based on fractional 
quantum Hall states~\cite{freedman2}, a theoretical formalism for QP detection 
needs to be developed. In Majorana tetrons, constructing a QP detection operator
was relatively straightforward, as Eq.~(\ref{KitaevModel}) is quadratic in 
field operators; equivalently, the effective (mean-field) Hamiltonian is 
non-interacting. Modeling QP detection in other systems is a problem of 
interest for topological quantum computation in general and also for 
fundamental physics. Similarly, the exponential localization of the WQP 
operators requires 
only that the many-body system has a bulk energy gap that remains open in the 
thermodynamic limit. The WQPs in Majorana tetrons can be used, in principle, to 
detect errors mainly because any even-parity error that acts on a MZM affects 
the state of a QP. It remains to be investigated if an analogous 
construction exists in other topological phases characterized by gapless 
surface modes, for example in fractional quantum Hall states. On the quantum 
error correction front, our work motivates inquiry into the practical 
implementation of active quantum error correction given limited control, and
finite detector resolution, of the microscopic degrees of freedom.

\begin{acknowledgments}
This research was supported by the Natural Sciences and
Engineering Research Council of Canada, the Alberta Major Innovation Fund
and the Australian Research
Council Centre of Excellence for Engineered Quantum Systems (CE170100009). 
A.A. acknowledges support by Killam Trusts (Postdoctoral Fellowship).
K.D.S. acknowledge support by the Simons Foundation (Postdoctoral Fellowship).
A.A. is grateful for insightful conversations with Salini Karuvade, 
Susan Coppersmith, Dominic Williamson, Andrew C. Doherty and Stephen Bartlett.
\end{acknowledgments}

\bibliography{References2}

\pagebreak

\appendix

\section{Properties of Exponentially Localized Operators}
\label{sec:app}
In this section, we provide some results on exponentially localized 
operators, which we later use for proving exponential localization of
WQPs. We begin these technical estimates with a definition.
\begin{definition}
\label{def:exploc}
    We say an operator $A$ on $\ell^2(\Z^+ \times \C^{D})$ is \textit{exponentially localized operator} with rate $\kappa > 0$ and prefactor~$C$ if
    \begin{equation}
    |\bra{j,m} A \ket{j',m}| \leq C e^{-\kappa |j - j'|}.
    \end{equation}
\end{definition}
Clearly, by definition, the Hamiltonian $H$ is an exponentially localized operator and the main content of Lemma~\ref{lem:p-exp-loc} is to prove that $P_{\rm qp}$ is  exponentially localized as well. Exponentially localized operators satisfy a number of important estimates which we will make great use of in our proofs.

The decay properties of exponentially localized operators are most easily expressed in terms of the exponential growth operator which depends on two real constants $g, x$: 
\begin{equation}
    B_{g,x} = \sum_{m,j} e^{g | j - x |} \ket{j,m} \bra{j,m}.
\end{equation}
In particular, exponentially localized operators have the following properties:
\begin{lemma}
\label{lem:exp-localization-estimates}
    If $A$ is an exponentially localized operator with rate $\kappa$ and prefactor $C$ then
    \begin{enumerate}
        \item $A$ is a bounded operator on $\ell^2(\Z^+ \times \C^{D})$.
        \item For all $x \in \R$ and all $0 \leq g \leq \frac{1}{2} \kappa$,
        \begin{equation}
        \| B_{g,x} A B_{g,x}^{-1} - A \| \leq 8 C D \kappa^{-2} g 
        \end{equation}
        \item For all $x \in \R$ and all $0 \leq g \leq \frac{1}{2} \kappa$,
        \begin{equation}
        \| [ B_{g,x} A B_{g,x}^{-1}, \tilde{X} ]_- \| \leq 8 C D \kappa^{-2},
        \end{equation}
    where $\tilde{X}=\sum_{m=1}^{D}\sum_j j\ket{j,m}\bra{j,m}$.
    \end{enumerate}
\end{lemma}
\noindent Note that $\tilde{X}=NX$ for the Kitaev chain, see the proof of Thm.~\ref{thm:qpposition}. 
\begin{remark}
The restriction $g \leq \frac{1}{2} \kappa$ is not strictly necessary. 
$g$ may be chosen to be any value strictly less than $\kappa$. 
However, the prefactor becomes larger as $|\kappa - g|$ becomes smaller. 
\end{remark}
The main technical lemma to prove Lemma \ref{lem:exp-localization-estimates} will be Schur's lemma specialized for linear operators on $\ell^2(\Z^+ \times \C^D)$:
\begin{theorem}[Schur's Lemma (Appendix I\cite{Grafakos2009})]
\label{thm:schurs-test}
Suppose $T$ is a linear operator on $\ell^2(\Z^+ \times \C^D)$ of the form
\begin{equation}
T = \sum_{j,m} \sum_{j',m'} \bra{j,m} T \ket{j',m'} \ket{j,m} \bra{j',m'}.
\end{equation}
That is $T$ is a discrete integral operator (infinite matrix) on $\ell^2(\Z^+ \times \C^D)$ with kernel function $K((j,m), (j', m')) := \bra{j,m} T \ket{j',m'}$.
If  
\begin{align}
\begin{split}
B_0 := \sup_{j,m} \sum_{j',m'} |\bra{j,m} T \ket{j',m'}|,\\
B_1 := \sup_{j',m'} \sum_{j,m} |\bra{j,m} T \ket{j',m'}|
\end{split}
\end{align}
are finite, then $\| T \| \leq \sqrt{B_0 B_1}$.
\end{theorem}
\noindent Here and henceforth, $:=$ stands for definition.

\begin{widetext}
\begin{proof}[Proof of Lemma \ref{lem:exp-localization-estimates}]
We start by writing $A$ in the position basis
\begin{equation}
A = \sum_{m,j} \sum_{m',j'} \bra{j,m} A \ket{j',m'} \ket{j',m'} \bra{j,m}.
\end{equation}
To show $A$ is bounded, following Schur's test, we first bound the ``$B_0$'' constant by
\begin{align}
\sup_{j, m} \sum_{m'=1}^D \sum_{j' \in \Z^+} | \bra{j,m} A \ket{j',m'} | 
& \leq \sup_{j, m} \sum_{m'=1}^D \sum_{j' \in \Z^+}  C e^{-\kappa | j - j' |} \nonumber\\
& \leq \sup_{j} \sum_{j' \in \Z^+}  CD e^{-\kappa | j - j' |}  \nonumber\\
& \leq \sup_{j} \int_{-\infty}^\infty CD e^{-\kappa |x - j|} \dd{x} \nonumber\\ 
& \leq 2 C D \kappa^{-1}. 
\end{align}
Repeating the same steps for the ``$B_1$'' constant in Schur's test, we find the same upper bound and hence $\| A \| \leq C D \kappa^{-1}$.

Now we proceed to show the bound for $B_{g,x} A B_{g,x}^{-1} - A$. Expanding this operator in the position basis we have:
\begin{equation}
B_{g,x} A B_{g,x}^{-1} - A = \sum_{m,j} \sum_{m',j'} \Big( e^{g |j' - x|} e^{-g |j - x|} - 1 \Big)  \bra{j,m} A \ket{j',m'} \ket{j',m'} \bra{j,m}
\end{equation}
Therefore, following Schur's test we consider
\begin{equation}
\sup_{j, m} \sum_{m'=1}^D \sum_{j' \in \Z^+}  |\Big( e^{g |j' - x|} e^{-g |j - x|} - 1 \Big)  \bra{j,m} A \ket{j',m'}|
\end{equation}
One can easily verify that for all $g > 0$ and all $j, j' \in \Z^+$
\begin{equation}
     | e^{g (|j' - x| - |j - x|)} - 1 | \leq | e^{g | j ' - j |} - 1 | \leq  g | j - j' | e^{g | j ' - j |}
\end{equation}

Hence so long as $g < \kappa$:
\begin{align}
\sup_{j, m} \sum_{m'=1}^D  \sum_{j' \in \Z^+}  |\Big( e^{g |j' - x|} e^{-g |j - x|} - 1 \Big)  \bra{j,m} A \ket{j',m'}|  &\leq \sup_{j, m} \sum_{m'=1}^D \sum_{j' \in \Z^+}  g | j - j' | e^{g | j ' - j |} |\bra{j,m} A \ket{j',m'}| \nonumber\\
& \leq \sup_{j} \sum_{j' \in \Z^+} C g D | j - j' | e^{(g - \kappa) | j ' - j |} \nonumber\\
& \leq 2 C  D g(\kappa - g)^{-2} \nonumber\\
& \leq 8 C D g\kappa^{-2}
\end{align}
where in the last line we have used that by assumption $g \leq \frac{1}{2} \kappa$. As before, a similar calculation shows that the ``$B_1$'' term can be bounded by the same quantity. Hence, 
\begin{equation}
\| B_{g,x} A B_{g,x}^{-1} - A  \| \leq 8 C D \kappa^{-2} g.
\end{equation}

For the last part of the lemma we have that
\begin{equation}
[B_{g,x} A B_{g,x}^{-1}, \tilde{X}]_- = \sum_{m,j} \sum_{m',j'} (j' - j) e^{g |j' - x|} e^{-g |j - x|}  \bra{j,m} A \ket{j',m'} \ket{j,m} \bra{j',m'}
\end{equation}
but it is easily verified that
\begin{equation}
| (j' - j)  e^{g |j' - x|} e^{-g |j - x|} | \leq | j - j' | e^{g |j - j'|}
\end{equation}
and hence the result follows by using an analogous calculation as for $B_{g,x} A B_{g,x}^{-1} - A$.
\end{proof}
\end{widetext}

\section{Proof of Lemma \ref{lem:p-exp-loc}}
\label{sec:p-exp-loc}
To prove this lemma, we will show that there exists a constant $C$ and a $g > 0$ so that for any $x \in \R$ the following estimate holds:
\begin{equation}
\label{eq:p-exp-tilt}
\| B_{g,x} P_{\rm qp} B_{g,x}^{-1} \| \leq C
\end{equation}
Proving Eq. \eqref{eq:p-exp-tilt} implies the lemma since by definition of the spectral norm:
\begin{equation}
\| A \| = \sup_{\| f \| = \| g \| = 1} \bra{g} A \ket{f}
\end{equation}
Hence, Eq. \eqref{eq:p-exp-tilt} implies
\begin{equation}
\begin{split}
 | \bra{j, m} B_{g,x} P_{\rm qp} B_{g,x}^{-1} \ket{j', m'} | &\leq C \\
 \Longrightarrow e^{g | j - x | } e^{- g | j' - x |} | \bra{j,m} P_{\rm qp} \ket{j',m'}| &\leq C.
\end{split}
\end{equation}
Since this holds for any $x \in \R$ we can choose $x = j'$ to conclude that
\begin{equation}
 | \bra{j,m} P_{\rm qp} \ket{j',m'}| \leq C e^{-g | j - j' |}
\end{equation}
which is what we wanted to show.

By the Riesz projection formula, we can express $P_{\rm qp}$ 
\begin{equation}
\label{eq:riesz-p}
P_{\rm qp} = \frac{1}{2 \pi i} \int_{\cC} (z - H)^{-1} \dd{z}
\end{equation}
where $\cC$ is a contour in the complex plane which encloses $\sigma(H) \cap [\Delta, \infty)$.
Note that since $H$ is bounded, $\cC$ has finite length.
We can formally calculate
\begin{align}
B_{g,x} P_{\rm qp} B_{g,x}^{-1} & = \frac{1}{2 \pi i} \int_{\cC} B_{g,x} (z - H)^{-1} B_{g,x}^{-1} \dd{z} \nonumber\\
& = \frac{1}{2 \pi i} \int_{\cC} \left(B_{g,x} (z - H) B_{g,x}^{-1} \right)^{-1} \dd{z} \nonumber\\
& = \frac{1}{2 \pi i} \int_{\cC} \left( z - B_{g,x} H B_{g,x}^{-1} \right)^{-1} \dd{z}.
\end{align}
These formal calculations are justified so long as we can show that the resolvent $( z - B_{g,x} H B_{g,x}^{-1})^{-1}$ is bounded for all $z \in \cC$. For any fixed $z$ we have that
\begin{align}
\begin{split}
    (z - B_{g,x} H B_{g,x}^{-1})^{-1}
    = (z - H - (B_{g,x} H B_{g,x}^{-1} - H) )^{-1} \\
= (z - H)^{-1} \Big(\mathds{1} - (B_{g,x} H B_{g,x}^{-1} - H) (z - H)^{-1} \Big)^{-1}
\end{split}
\end{align}
Hence, if $z \not\in \sigma(H)$ and $\| (B_{g,x} H B_{g,x}^{-1} - H) (z - H)^{-1} \| < 1$ we can conclude that $(z - B_{g,x} H B_{g,x}^{-1})^{-1}$ is bounded.

In particular, if we define $\delta^{-1} := \sup_{z \in \cC} \| (z - H)^{-1} \|$ then Lemma \ref{lem:exp-localization-estimates} implies for all $z \in \cC$
\begin{equation}
\| (B_{g,x} H B_{g,x}^{-1} - H) (z - H)^{-1} \| \leq 8 C D \delta^{-1} g
\end{equation}
Hence if $g < \frac{\delta}{8 C D}$, the above calculation implies for all $z \in \cC$
\begin{equation}
\| (z - B_{g,x} H B_{g,x}^{-1})^{-1} \| \leq \frac{\delta^{-1}}{1 - 8 C D \delta^{-1} g}.
\end{equation}
This finally implies that
\begin{equation}
\| B_{g,x} P_{\rm qp} B_{g,x}^{-1} \| \leq \frac{\ell(\cC)}{2\pi} \frac{\delta^{-1}}{1 - 8 C D \delta^{-1} g}
\end{equation}
where $\ell(\cC)$ is the length of the contour $\cC$ in the complex plane. 

\begin{remark}
    From these calculations, we estimate the exponential decay rate of $P_{\rm qp}$ is $O(\delta^{-1})$ (i.e. the decay rate of $P_{\rm qp}$ is proportional to the inverse gap of $H$). 
\end{remark}

\section{Proof of Theorem \ref{thm:wannier-1d}}
\label{sec:wannier-1d}
For this proof we will use Lemma \ref{lem:p-exp-loc} and suppose $P_{\rm qp}$ is an exponentially localized operator with rate $\kappa_{\rm qp}$ and prefactor $C_{\rm qp}$. We will also make use of the following simple lemma which has appeared in a number of previous works~\cite{2007Kittaneh,2008WangDu,lu2022existence}
\begin{lemma}
  \label{lem:pos-comm-bd}
  Let $A_1,A_2$ be two bounded operators. If $A_1$ is positive semidefinite then
  \begin{equation}
    \|[A_1,A_2]_-\| \leq \|A_1\|  \|A_2\|.
  \end{equation}
  If both $A_1$ and $A_2$ are positive semidefinite then
  \begin{equation}
    \|[A_1,A_2]_-\| \leq \frac{1}{2} \|A_1\|  \|A_2\|.
  \end{equation}
\end{lemma}
\begin{proof}
    Let's first suppose that only $A_1$ is positive semidefinite. Since $A_1$ is bounded, we have that $\sigma(A_1) \subseteq [0, \| A_1 \|]$. Now define the operator $\tilde{A}_1 := A_1 - 
    \frac{1}{2} \| A_1 \|$ and observe that 
    \begin{equation}
    \sigma( \tilde{A}_1 ) \subseteq \left[ -\frac{1}{2} \|A_1\|, \frac{1}{2} \| A_1 \| \right].
    \end{equation}
    Hence $\| \tilde{A}_1 \| = \frac{1}{2} \| A_1 \|$. Therefore
    \begin{multline}
    \|[A_1,A_2]_-\| = \|[\tilde{A}_1,A_2]_-\| \\
    \leq 2 \| \tilde{A}_1 \| \| A_2 \| 
    = \| A_1 \| \| A_2 \|.
    \end{multline}
    If $A_2$ is also positive semidefinite, we can replace $A_2$ with $\tilde{A}_2 := A_2 - \frac{1}{2} \| A_2 \|$ as well to get the second bound.
\end{proof}

\begin{widetext}
\subsection{Proof $P_{\rm qp} \tilde{X} P_{\rm qp}$ has discrete spectrum}
We will prove that $P_{\rm qp} \tilde{X} P_{\rm qp}$ has compact resolvent on $\range{(P_{\rm qp})}$ and hence $P_{\rm qp} \tilde{X} P_{\rm qp}$ only has discrete eigenvalues. In the first step, we show that for sufficiently large $\alpha>0$, the operator $P_{\rm qp} \tilde{X} P_{\rm qp} \pm i \alpha$ has full rank on $\range{(P_{\rm qp})}$.
Consider the following expression
\begin{equation}
(P_{\rm qp} \tilde{X} P_{\rm qp} \pm i \alpha) 
P_{\rm qp} ( \tilde{X} \pm i \alpha )^{-1} P_{\rm qp}.
\end{equation}
Note that $(\tilde{X} \pm i \alpha)^{-1}$ is well defined and $\| (\tilde{X} \pm i \alpha)^{-1} \| \leq \alpha^{-1}$. 
Since $P_{\rm qp}$ is a projection we have
\begin{equation}
\begin{split}
   (P_{\rm qp} \tilde{X} P_{\rm qp} \pm i \alpha) P_{\rm qp} ( \tilde{X} \pm i \alpha )^{-1} P_{\rm qp} 
   & = P_{\rm qp} (\tilde{X}  \pm i \alpha) P_{\rm qp} ( \tilde{X} \pm i \alpha )^{-1} P_{\rm qp} \\
   & = P_{\rm qp} \Big(P_{\rm qp}(\tilde{X}  \pm i \alpha) + [\tilde{X}  \pm i \alpha,P_{\rm qp}]_-\Big) 
   ( \tilde{X} \pm i \alpha )^{-1} P_{\rm qp} \\
   & = P_{\rm qp} + P_{\rm qp} [ \tilde{X}, P_{\rm qp} ]_- ( \tilde{X} \pm i \alpha )^{-1} P_{\rm qp}.
\end{split}
\end{equation}
Since $[ \tilde{X}, P_{\rm qp} ]_-$ is bounded by Lemma~\ref{lem:exp-localization-estimates}, by choosing $\alpha$ sufficiently large we can ensure that $\| [ \tilde{X}, P_{\rm qp} ]_- ( \tilde{X} \pm i \alpha )^{-1} \| < 1$ which implies that $(P_{\rm qp} \tilde{X} P_{\rm qp} \pm i \alpha)$ is full rank on $\range{(P_{\rm qp})}$ and hence self-adjoint on this domain\cite[Theorem VIII.3]{reed1980functional}.

Consequently, we can invert $(P_{\rm qp} \tilde{X} P_{\rm qp} \pm i \alpha)$ and $P_{\rm qp} + P_{\rm qp} [ \tilde{X}, P_{\rm qp} ]_- ( \tilde{X} \pm i \alpha )^{-1} P_{\rm qp}$ to get that
\begin{equation}
\begin{split}
 (P_{\rm qp} \tilde{X} P_{\rm qp} \pm i \alpha)^{-1} &= P_{\rm qp} ( \tilde{X} \pm i \alpha )^{-1} P_{\rm qp} \Big( P_{\rm qp} + P_{\rm qp} [ \tilde{X}, P_{\rm qp} ]_- ( \tilde{X} \pm i \alpha )^{-1} P_{\rm qp}\Big)^{-1}\\
 &= \left(P_{\rm qp} ( \tilde{X} \pm i \alpha )^{-1} \right)\Big( P_{\rm qp} + P_{\rm qp} [ \tilde{X}, P_{\rm qp} ]_- ( \tilde{X} \pm i \alpha )^{-1} P_{\rm qp}\Big)^{-1}.
\end{split}
\end{equation}
We now show that the first factor in the final expression is compact.
For any integer $j_0 > 0$ we can write $P_{\rm qp} ( \tilde{X} \pm i \alpha )^{-1}$ as
\begin{equation}
P_{\rm qp} ( \tilde{X} \pm i \alpha )^{-1} = \Big(  P_{\rm qp} \chi(|\tilde{X}| \leq j_0) \Big) ( \tilde{X} \pm i \alpha )^{-1} + P_{\rm qp} \Big( \chi(|\tilde{X}| > j_0) ( \tilde{X} \pm i \alpha )^{-1} \Big),
\end{equation}
where
\begin{equation}
\chi(|\tilde{X}| \leq j_0) := \sum_{m=1}^D \sum_{|j| \leq j_0} \ket{j,m} \bra{j,m}=:\mathds{1}-\chi(|\tilde{X}| > j_0).
\end{equation}
\end{widetext}
But observe that
\begin{align}
\| P_{\rm qp} \chi(|\tilde{X}| \leq j_0) \|_{F}^2 
& = \sum_{m,j} \sum_{m', |j| \leq j_0} |\bra{j,m} P_{\rm qp} \ket{j',m'} |^2 \nonumber\\
& \leq 2 j_0 D \Big( \sup_{j'} \sum_{m,j} |\bra{j,m} P_{\rm qp} \ket{j',m'} |^2 \Big) \nonumber\\
& \leq 2 C_{\rm qp} j_0 D^2 \sum_{j} e^{-\kappa_{\rm qp} |j|} < \infty.
\end{align}
Therefore $P_{\rm qp} \chi(|\tilde{X}| \leq j_0)$ is a Hilbert-Schmidt operator and hence compact. By definition  
\begin{align}
\| P_{\rm qp} \chi(|\tilde{X}| > j_0) ( \tilde{X} \pm i \alpha )^{-1} \| &\leq \| \chi(|\tilde{X}| > j_0) ( \tilde{X} \pm i \alpha )^{-1} \| \nonumber\\
&\leq (j_0^2 + \alpha^2)^{-1/2}.
\end{align}
Hence, we have that $P_{\rm qp} ( \tilde{X} \pm i \alpha )^{-1}$ is the limit of compact operators in the operator norm and hence also compact. Therefore, $(P_{\rm qp} \tilde{X} P_{\rm qp} \pm i \alpha)^{-1}$ is compact and $P_{\rm qp} \tilde{X} P_{\rm qp}$ has compact resolvent on $\range{(P_{\rm qp})}$.

\subsection{Proof $P_{\rm qp} \tilde{X} P_{\rm qp}$ has exponentially localized eigenfunctions}
For this proof we will introduce notation for step function
\begin{equation}
\chi_{x, b} := \sum_{m=1}^D \sum_{|j - x| \leq b} \ket{j,m} \bra{j,m},
\end{equation}
which depends on parameters $x \in \R$, $b \in \R^+$.
That is $\chi_{x, b}$ is a ``step'' of width $b$ centered at the real number $x$. 
Observe that for any choice of $x$ or $b$ the operator $\chi_{x,b}$ is a positive operator. With this definition, we now proceed with the proof.

Suppose that $\ket{\phi} \in \range{(P_{\rm qp})}$ is an eigenfunction of $P_{\rm qp} \tilde{X} P_{\rm qp}$ with eigenvalue $x$. By definition this means
\begin{equation}
P_{\rm qp} \tilde{X} P_{\rm qp} \ket{\phi} = x \ket{\phi}\Longleftrightarrow P_{\rm qp} (\tilde{X} - x) P_{\rm qp} \ket{\phi} = 0.
\end{equation}
Let $b \in \R^+$ be a parameter to be chosen later, adding $i b \chi_{x, b} \ket{\phi}$ to both sides gives
\begin{equation}
\label{eq:pxp-proof-1}
P_{\rm qp} (\tilde{X} - x + i b \chi_{x,b}) P_{\rm qp} \ket{\phi} = i b \chi_{x,b} \ket{\phi}.
\end{equation}
Our basic goal in this proof will be to rewrite Eq. \eqref{eq:pxp-proof-1} as
\begin{equation}
\label{eq:pxp-proof-2}
\ket{\phi} = \mathcal{L} \chi_{x,b} \ket{\phi},
\end{equation}
where $\mathcal{L}$ is some linear operator. Once we can show this, we will multiply both sides by $B_{g,x}$ to get that
\begin{equation}
B_{g,x} \ket{\phi} = B_{g,x} \mathcal{L} \chi_{x,b} \ket{\phi} = \left( B_{g,x} \mathcal{L} B_{g,x}^{-1} \right) \left( B_{g,x} \chi_{x,b} \ket{\phi} \right).
\end{equation}
We will then show that for a choice of $b$ large enough and all $g$ sufficiently small, there exists a  constant $C$ so that $\| B_{g,x} \mathcal{L} B_{g,x}^{-1} \| \leq C$. Since $\chi_{x,b}$ is a step function we have 
\begin{equation}
\| B_{g,x} \ket{\phi} \| \leq C e^{g b} \| \ket{\phi} \|,
\end{equation}
which immediately implies that
\begin{equation}
| \bra{j, m} \ketr{\phi}| \leq C e^{g b} e^{-g | j - x|}.
\end{equation}
Importantly, the choice of $b$ depends only on $P_{\rm qp}$ and not on $j$ or $\ket{\phi}$. Hence the theorem is proved.

Returning to Eq.~\eqref{eq:pxp-proof-1}, since $\tilde{X} - x$ has a zero at $x$ and $\chi_{x,b}$ is a step function centered at $x$, so long as $b > 0$, the operator $(\tilde{X} - x + i b \chi_{x,b})$ is invertible and
\begin{equation}
\| (\tilde{X} - x + i b \chi_{x,b})^{-1} \| \leq b^{-1}.
\end{equation}
Therefore, multiplying both sides of Eq. \eqref{eq:pxp-proof-1} by $(\tilde{X} - x + i b \chi_{x,b})^{-1}$ gives
\begin{widetext}
\begin{equation}
(\tilde{X} - x + i b \chi_{x,b})^{-1} P_{\rm qp} (\tilde{X} - x + i b \chi_{x,b}) P_{\rm qp} \ket{\phi} = i (\tilde{X} - x + i b \chi_{x,b})^{-1} b \chi_{x,b} \ket{\phi}.
\end{equation}
Commuting $P_{\rm qp}$ and $(\tilde{X} - x + i b \chi_{x,b})$ on the left hand side to get
\begin{equation}
\Big(\mathds{1} + (\tilde{X} - x + i b \chi_{x,b})^{-1} \big([P_{\rm qp}, \tilde{X}]_- + i [P_{\rm qp}, b \chi_{x,b}]_-\big) \Big) P_{\rm qp} \ket{\phi} = i (\tilde{X} - x + i b \chi_{x,b})^{-1} b \chi_{x,b} \ket{\phi}.
\end{equation}
Upon examining the left hand side a bit more closely, we see that 
\begin{equation}
\begin{split}
\| (\tilde{X} - x + i b \chi_{x,b})^{-1} & \big([P_{\rm qp}, \tilde{X}]_- + i [P_{\rm qp}, b \chi_{x,b}]_-\big) \| \\[1ex]
& \leq \| (\tilde{X} - x + i b \chi_{x,b})^{-1} \| \left( \| [P_{\rm qp}, \tilde{X}]_- \| + b \| [P_{\rm qp}, \chi_{x,b}]_- \| \right) \\[1ex]
& \leq b^{-1} \left( \| [P_{\rm qp}, \tilde{X}]_- \| + \frac{1}{2} b \right),
\end{split}
\end{equation}
where in the last line we have used \ref{lem:pos-comm-bd}. 
Since 
\begin{equation}
\| [P_{\rm qp}, \tilde{X}]_- \| \leq 8 C_{\rm qp} d \kappa_{\rm qp}^{-2}
\end{equation}
by Lemma \ref{lem:exp-localization-estimates},
by choosing $b = 32 C_{\rm qp} d \kappa_{\rm qp}^{-2}$, we can ensure that 
\begin{equation}
\| (\tilde{X} - x + i b \chi_{x,b})^{-1} \big([P_{\rm qp}, \tilde{X}]_- + i [P_{\rm qp}, b \chi_{x,b}]_-\big) \| \leq \frac{3}{4}.
\end{equation}
In this case, we can invert the operator on the left hand side to get 
\begin{equation}
\ket{\phi} = i b \Big(\mathds{1} + (\tilde{X} - x + i b \chi_{x,b})^{-1} \big([P_{\rm qp}, \tilde{X}]_- + i [P_{\rm qp}, b \chi_{x,b}]_-\big) \Big)^{-1} (\tilde{X} - x + i b \chi_{x,b})^{-1} \chi_{x,b} \ket{\phi}.
\end{equation}
We have now recovered Eq. \eqref{eq:pxp-proof-2} where
\begin{equation}
\mathcal{L} = ib \Big(\mathds{1} + (\tilde{X} - x + i b \chi_{x,b})^{-1} \big([P_{\rm qp}, \tilde{X}]_- + i [P_{\rm qp}, b \chi_{x,b}]_-\big) \Big)^{-1} (\tilde{X} - x + i b \chi_{x,b})^{-1}.
\end{equation}
Following the argument given above, to complete the proof we only need to show that there exists a constant $C$ so that $\| B_{g,x} \mathcal{L} B_{g,x}^{-1} \| \leq C$. Since $\tilde{X}$ and $\chi_{x,b}$ are both diagonal in the position basis we see that
\begin{equation}
\begin{split}
B_{g,x}& \mathcal{L} B_{g,x}^{-1} = ib \Big( \mathds{1} + (\tilde{X} - x + i b \chi_{x,b})^{-1} \big([ B_{g,x} P_{\rm qp} B_{g,x}^{-1}, \tilde{X}]_- + i [ B_{g,x} P_{\rm qp} B_{g,x}^{-1}, b \chi_{x,b}]_-\big) \Big)^{-1} (\tilde{X} - x + i b \chi_{x,b})^{-1}.
\end{split}
\end{equation}
Therefore, since $\|  (\tilde{X} - x + i b \chi_{x,b})^{-1} \| \leq b^{-1}$,
\begin{equation}
\| B_{g,x} \mathcal{L} B_{g,x}^{-1} \| \leq \| \Big( \mathds{1} + (\tilde{X} - x + i b \chi_{x,b})^{-1} \big([ B_{g,x} P_{\rm qp} B_{g,x}^{-1}, \tilde{X}]_- + i [ B_{g,x} P_{\rm qp} B_{g,x}^{-1}, b \chi_{x,b}]_-\big) \Big)^{-1} \|,
\end{equation}
and hence to complete the proof it suffices to show that there exists a choice of $b$ so that for all $g$ sufficiently small
\begin{equation}
\label{eq:pxp-inv-2}
\|  (\tilde{X} - x + i b \chi_{x,b})^{-1} \big([ B_{g,x} P_{\rm qp} B_{g,x}^{-1}, \tilde{X}]_- + i [ B_{g,x} P_{\rm qp} B_{g,x}^{-1}, b \chi_{x,b}]_-\big) \| < 1.
\end{equation}
We calculate
\begin{equation}
\begin{split}
   \|  (\tilde{X} - x + i b \chi_{x,b})^{-1} & \big([ B_{g,x} P_{\rm qp} B_{g,x}^{-1}, \tilde{X}]_- + i [ B_{g,x} P_{\rm qp} B_{g,x}^{-1}, b \chi_{x,b}]_-\big) \| \\[1ex]
   & \leq b^{-1} \left( \| [ B_{g,x} P_{\rm qp} B_{g,x}^{-1}, \tilde{X}]_- \| + \| [ B_{g,x} P_{\rm qp} B_{g,x}^{-1} - P_{\rm pq} + P_{\rm pq}, b \chi_{x,b}]_- \|  \right) \\[1ex]
   & \leq b^{-1} \| [ B_{g,x} P_{\rm qp} B_{g,x}^{-1}, \tilde{X}]_- \| + \| [ B_{g,x} P_{\rm qp} B_{g,x}^{-1} - P_{\rm pq}, \chi_{x,b}]_- \| + \| [ P_{\rm pq}, \chi_{x,b}]_- \| \\
   & \leq b^{-1} \| [ B_{g,x} P_{\rm qp} B_{g,x}^{-1}, \tilde{X}]_- \| + \| B_{g,x} P_{\rm qp} B_{g,x}^{-1} - P_{\rm pq} \| + \frac{1}{2}, \\
\end{split}
\end{equation}
\end{widetext}
where in the last line we have used Lemma \ref{lem:pos-comm-bd} twice. By Lemma \ref{lem:exp-localization-estimates},
\begin{equation}
\begin{split}
    & \| [ B_{g,x} P_{\rm qp} B_{g,x}^{-1}, \tilde{X}]_- \|  \leq 8 C_{\rm qp} d \kappa_{\rm qp}^{-2}, \\
    & \| B_{g,x} P_{\rm qp} B_{g,x}^{-1} - P_{\rm pq} \| 
    \leq 8 C_{\rm qp} d \kappa_{\rm qp}^{-2} g. \\
\end{split}
\end{equation}
Therefore, if we choose $b = 64 C_{\rm qp} D \kappa_{\rm qp}^{-2}$ and $g = \kappa_{\rm qp}^{2}(64 C_{\rm qp} D)^{-1}$ we conclude that $\| B_{g,x} \mathcal{L} B_{g,x}^{-1} \| \leq 4$. Therefore, following the previous argument, we conclude that 
\begin{equation}
| \bra{j, m} \ketr{\phi}| \leq 4 e \exp\left(-\frac{\kappa_{\rm qp}^{2}}{64 C_{\rm qp} D} | j - x|\right).
\end{equation}
Note that for the Kitaev chain Hamiltonian, the eigenstates of $X_{\rm qp}$ with 
positive eigenvalues are eigenstates of $P_{\rm qp} \tilde{X} P_{\rm qp}$, so 
this finishes the proof of Thm~\ref{thm:wannier-1d}.

\section{Average speed of the quasiparticles}
\label{app:qpvelocity}

In this section, we derive an expression for the average speed of QPs
in a superconductor in thermal equilibrium. For simplicity, let us consider
a one-dimensional conductor with dispersion relation
$\xi(k) = \hbar^2 k^2/2m$ and with Fermi momentum $k_{\rm F}$.  
The dispersion relation after the onset of
superconductivity is given by $E(k) = \pm\sqrt{\xi(k)^2 + \Delta^2}$,
where $\Delta$ is the pairing strength. 
In the vicinity of $k_{\rm F}$, the energy $\xi(k)$ can be assumed to be linear,
i.e. $\xi(k) = \hbar v_F(k-k_{\rm F})$ for $|k-k_{\rm F}| \ll |k_{\rm F}|$
where $v_{\rm F} = \hbar k_{\rm F}/m$. 
Moreover, for $|k-k_{\rm F}| \ll \Delta/\hbar v_F$, $E(k)$ can be approximated
by
\begin{equation}
E(k) =  \Delta + \frac{\hbar^2v_F^2 (k-k_{\rm F})^2}{2\Delta} = 
\Delta + \frac{\hbar^2 (k-k_{\rm F})^2}{2m_{\rm qp}},
\end{equation}
where $m_{\rm qp}$ is the effective mass of the QPs given by
\begin{equation}
\label{eq:effectivemass}
\frac{1}{m_{\rm qp}} = \frac{1}{\hbar^2}\left.\frac{\partial^2 E(k)}{\partial k^2}\right\vert_{k_{\rm F}} =
 \frac{v_F^2}{ \Delta}.
\end{equation}
The QPs are governed by the Fermi-Dirac distribution and
their average occupation is given by~\cite{Guadagnini2017} 
\begin{equation}
\expval{n_k} = \frac{1}{1+e^{\beta E(k)}} .
\end{equation} 
For $\beta\Delta \gg 1$, we have 
$\beta E(k) = \beta\Delta + \beta\hbar^2v_F^2 (k-k_{\rm F})^2/2\Delta \gg 1$,
and therefore
\begin{equation}
\expval{n_k} \approx e^{-\beta E(k)} = 
e^{-\beta\Delta}e^{-\beta \hbar^2(k-k_{\rm F})^2/2m_{\rm qp}}.
\end{equation}
The fraction of QPs with quasi-momentum $k$ is then given by
\begin{equation}
\label{eq:idealgas}
\frac{\expval{n_k}}{\sum_k\expval{n_k}}  \approx 
\frac{e^{-\beta \hbar^2(k-k_{\rm F})^2/2m_{\rm qp}}}{\sum_k e^{-\beta \hbar^2(k-k_{\rm F})^2/2m_{\rm qp}}}.
\end{equation}
Eq.~\eqref{eq:idealgas} reflects the well-known fact that Fermi-Dirac distribution
converges to Maxwell-Boltzmann distribution in the limit of low particle density.
The distribution of QPs is thus identical to that of a 1D classical ideal gas,
and therefore we can directly use the corresponding formula for the average speed
$\expval{|v_{\rm qp}|} \sim \sqrt{1/\beta m_{\rm qp}}$. Substituting the
expression for $m_{\rm qp}$ in Eq.~\eqref{eq:effectivemass}, 
we get
\begin{equation}
\expval{|v_{\rm qp}|} \sim v_F\sqrt{1/\beta \Delta}
\end{equation}
as desired.

\end{document}